\documentclass[sigconf, nonacm]{acmart}

\newif\ifarxiv
\arxivtrue 
\newcommand{\shortlong}[2]{{\ifarxiv{#2}\else{#1}\fi}}

\ifarxiv
\allowdisplaybreaks
\fi

\ifarxiv
\else

\newcommand\vldbdoi{XX.XX/XXX.XX}
\newcommand\vldbpages{XXX-XXX}
\newcommand\vldbvolume{14}
\newcommand\vldbissue{1}
\newcommand\vldbyear{2020}
\newcommand\vldbauthors{\authors}
\newcommand\vldbtitle{\shorttitle} 
\newcommand\vldbavailabilityurl{https://github.com/cmla-psu/QuerySmasherMatrixMechanism}
\newcommand\vldbpagestyle{plain} 

\fi

\usepackage{color}
\usepackage{xargs} 
\usepackage{xspace} 
\usepackage[createShortEnv,conf={restate,no link to proof}]{proof-at-the-end}

\usepackage[linesnumbered,ruled,vlined]{algorithm2e}
\usepackage{hyperref}

\usepackage{enumitem}
\usepackage{bbm}

\newtheorem{theorem}{Theorem}

\newtheorem{definition}{Definition}

\DeclareMathAlphabet{\mathdutchcal}{U}{dutchcal}{m}{n}
\SetMathAlphabet{\mathdutchcal}{bold}{U}{dutchcal}{b}{n}
\DeclareMathAlphabet{\mathdutchbcal}{U}{dutchcal}{b}{n}

\def\notationcolor{black} 
\newcommand{\notation}[2]{\newcommand{#1}{{\textcolor{\notationcolor}{\ensuremath{#2}}}}}

\newcommand{\term}[2]{\newcommand{#1}{\textcolor{\notationcolor}{#2}\xspace}}

\newcommand{\mat}[1]{\mathbf{#1}}
\newcommand{\bvec}[1]{\mathbf{#1}}

\newcommand{\mquery}[1]{\textcolor{\notationcolor}{Q^{(marg)}_{#1}}}

\notation{\outp}{\omega}
\notation{\resid}{\text{Resid}} 
\notation{\mech}{\mathcal{M}} 
\notation{\range}{range} %
\notation{\pcost}{pcost}
\notation{\randalg}{\mathcal{G}}
\notation{\weight}{\text{weight}}
\notation{\recon}{\text{Recon}}
\notation{\varfun}{\text{Var}}
\notation{\trace}{\text{trace}}
\notation{\pcostbound}{\beta}

\notation{\data}{\mathcal{D}} 
\notation{\datavec}{\vec{\mathbf{x}}}
\notation{\attr}{A} 
\notation{\attv}{a} 
\notation{\attsize}{d} 
\notation{\numattr}{k} 
\notation{\rec}{r} 
\notation{\hot}{\mathbbm{e}} 
\notation{\numrec}{n} 
\notation{\rcmatrix}{\mathfrak{R}}
\notation{\kron}{\otimes}
\notation{\Kron}{\bigotimes}
\notation{\marginal}{\bvec{T}}
\notation{\attset}{\mathcal{A}}
\notation{\imag}{\mathbbmtt{i}}
\notation{\basis}{\mathbf{\mathdutchcal{f}}}
\notation{\extend}{\text{extend}}
\notation{\mopt}{\text{MOpt}}

\notation{\covar}{\mat{\Sigma}}
\notation{\zero}{\bvec{0}}
\notation{\bmat}{\mat{B}}
\notation{\one}{\bvec{1}}
\notation{\identity}{\mat{\mathcal{I}}}

\notation{\query}{\bvec{q}}
\notation{\queryv}{\vec{\query}}
\notation{\Query}{\mat{Q}}
\notation{\workload}{\text{WKLoad}}
\notation{\support}{\text{support}}
\notation{\deco}{\Rightarrow}

\notation{\profile}{\mathcal{V}}
\notation{\loss}{\mathcal{L}}

\term{\longsysname}{Multi-basis Residual Planner}
\term{\sysname}{QuerySmasher}

\begin{document}
\title{Accurate and Scalable Matrix Mechanisms via Divide and Conquer}

\author{Guanlin He}
\affiliation{%
  \institution{Penn State University}
}
\email{gbh5146@psu.edu}

\author{Yingtai Xiao}
\affiliation{%
  \institution{TikTok}
}
\email{yingtai0323@gmail.com}

\author{Jiamu Bai}
\affiliation{%
  \institution{Penn State University}
}
\email{jvb6867@psu.edu}

\author{Xin Gu}
\affiliation{%
  \institution{Penn State University}
}
\email{xingu@psu.edu}

\author{Zeyu Ding}
\affiliation{%
  \institution{Binghamton University}
}
\email{dding1@binghamton.edu}


\author{Wenpeng Yin}
\affiliation{%
  \institution{Penn State University}
}
\email{wenpeng@psu.edu}

\author{Daniel Kifer}
\affiliation{%
  \institution{Penn State University}
}
\email{dkifer@cse.psu.edu}

\begin{abstract}
Matrix mechanisms are often used to provide unbiased differentially private
query answers when publishing statistics or creating synthetic data.
Recent work 
has developed matrix mechanisms, such as ResidualPlanner and Weighted Fourier Factorizations, 
that scale to high dimensional
datasets while providing optimality guarantees for workloads such as
marginals and circular product queries. They operate by adding noise
to a linearly independent set of queries that can compactly represent
the desired workloads.

In this paper, we present \sysname, an alternative scalable approach based on
a divide-and-conquer strategy. Given a workload that can be answered
from various data marginals, \sysname splits each query into subqueries
and re-assembles the pieces into mutually  orthogonal subworkloads.
These subworkloads represent
small, low-dimensional problems that can be independently and optimally answered by existing low-dimensional 
matrix mechanisms.
\sysname then stitches these solutions together to answer queries in the original workload.

We show that \sysname subsumes prior work, like ResidualPlanner (RP), ResidualPlanner+ (RP+),
and Weighted Fourier Factorizations (WFF). We prove that it can dominate those approaches, under sum squared error,
for all workloads. We also experimentally demonstrate the scalability and accuracy
of \sysname.

\end{abstract}

\maketitle


\section{Introduction}\label{sec:intro}
Confidential datasets about individuals and businesses provide a wealth of information
that is useful for scientific, economic, and policy purposes. Clearly, they cannot be
released without modification and so a variety of perturbed summary statistics are released
instead. The fundamental theorem of database reconstruction \cite{dinur:privacy}
demonstrates limits on how many approximate query answers can be released before
privacy breaks down under any reasonable notion, and differential privacy \cite{dwork06Calibrating,zcdp,renyidp,fdp}
has become a de facto standard for navigating the privacy/utility tradeoff.

Although the  datasets can be high-dimensional, the queries of interest \cite{census2020,qcewdata,acsdata} often involve a small set of variables at a time. Hence, the queries of interest can be answered from a collection of data marginals. For example, age range queries for each race can be answered by a 2-way marginal on age and race, 2-d range queries on age and income for each gender can be answered by a 3-way marginal on age, gender, and income, etc.

These are examples of \emph{linear queries}, and they can be answered using a class of differentially private mechanisms known as \emph{matrix mechanisms} \cite{LMHMR15,YZWXYH12,YYZH16,edmonds2020power,mckenna2018optimizing,rplanner,rplus,lebeda2025weightedfourierfactorizationsoptimal}. Given a workload of linear queries, these methods create a set of \emph{strategy queries} that can be answered with low noise. The original workload queries can then be answered by linear postprocessing of the strategy query answers. This is a powerful framework that has been used in the 2020 Census disclosure avoidance system \cite{tdahdsr} and is also a key element of more complex algorithms \cite{zhang2018ektelo,aim,privbayes,fuentes2026fast,mullins2024efficient,mckenna2025scaling}.

One of the key challenges of the matrix mechanism line of work is scalability to high-dimensional datasets. If a dataset has attributes $\attr_1,,\dots, \attr_\numattr$ with domains sizes $\attsize_1,\dots, \attsize_\numattr$, then optimizing a matrix mechanism to minimize the sum squared error of an \emph{arbitrary} query workload  can be formulated as a semi-definite programming problem with $O\left(\left(\prod_{i=1}^\numattr d_i\right)^2\right)$ variables \cite{YYZH16}. Solving such an optimization problem is only feasible for low-dimensional datasets that have a small number of attributes. 
However, most applications require answers to \emph{structured} query workloads rather than arbitrary workloads.
Specifically, each linear query accesses a small number of attributes (different queries can access different attributes) and so the entire workload can be computed
using a collection of lower dimensional marginals \cite{census2020,qcewdata,acsdata}.
McKenna  et al. \cite{mckenna2018optimizing} first showed how to scale matrix mechanism construction  by exploiting such structure. However, their algorithm still required $O\left(\prod_{i=1}^\numattr d_i\right)$ space, which becomes a bottleneck for scalability.

A recent line of work called ResidualPlanner (RP) \cite{rplanner}, ResidualPlanner+ (RP+) \cite{rplus}, and Weighted Fourier Factorizations (WFF) \cite{lebeda2025weightedfourierfactorizationsoptimal} overcame this bottleneck by proposing matrix mechanisms whose time and space complexity scales with the output size (i.e., size of the marginals needed to answer the workload queries) rather than the exponentially sized data domain. This is enabled by special types of linearly independent queries, such as the residual basis \cite{rplanner} or multidimensional Fourier basis \cite{lebeda2025weightedfourierfactorizationsoptimal}, that can compactly represent any collection of marginals and therefore they can compactly represent any linear workload that is computable from collections of marginals (such as range queries, prefix-sum queries, etc.).
Surprisingly, a proper choice of basis also guarantees optimality (among matrix mechanisms) for expressive classes of queries like marginals \cite{rplanner} or, more generally, circular product queries \cite{lebeda2025weightedfourierfactorizationsoptimal}.

In this paper, we present \sysname, a novel approach to matrix mechanism scalability that is provably more accurate than this prior work \cite{rplanner,rplus,lebeda2025weightedfourierfactorizationsoptimal} under the weighted sum squared error metric.
It breaks apart each query into subqueries and groups those pieces into subworkloads. These subworkloads are mutually orthogonal, so that matrix mechanisms for each subworkload can be constructed independently. The subworkloads are also
low dimensional, so that optimal low-dimensional techniques \cite{YYZH16} can be used to generate the matrix mechanisms for each subworkload. Afterwards, these matrix mechanisms are combined together to create a mechanism for the original workload.
This divide-and-conquer approach ensures scalability. \sysname can also be viewed from an alternate lens of automatically searching for a good workload-dependent compact linear basis that can represent marginals, instead of always using the residual basis \cite{rplanner} or Fourier basis \cite{lebeda2025weightedfourierfactorizationsoptimal}. This ensures high accuracy.

To summarize, our contributions are the following:
\begin{itemize}[leftmargin=*]
\item We propose a new matrix mechanism, called \sysname, for computing privacy-preserving answers to any workload of linear queries that are computable from sets of marginals. \sysname uses a novel query decomposition approach to achieve both scalability and high accuracy. 
\item We mathematically prove that \sysname can dominate the scalable compact-linear-basis approaches RP \cite{rplanner}, RP+ \cite{rplus}, and WFF \cite{lebeda2025weightedfourierfactorizationsoptimal}, in terms of weighted sum squared error, for every workload. Specifically, if the optimal low-dimensional matrix mechanism construction \cite{yuan2015optimizing} for each subworkload is replaced by approximate methods, the result becomes mathematically equivalent to RP, RP+, and WFF. 
\item We empirically evaluate \sysname against competing scalable matrix mechanisms  \cite{rplanner,rplus,lebeda2025weightedfourierfactorizationsoptimal,mckenna2018optimizing} to show that its theoretical accuracy and scalability advantages translate into practice.
\end{itemize}

This paper is organized as follows. In Section \ref{sec:notation}, we introduce our notation. In Section \ref{sec:background} we explain background on differential privacy and matrix mechanisms. In Section \ref{sec:related}, we discuss related work. We present an intuitive overview of \sysname in Section \ref{sec:overview} and the technical details in Section \ref{sec:decomp}. In Section \ref{sec:accuracy}, we prove that \sysname dominates RP \cite{rplanner}, RP+ \cite{rplus}, and WFF \cite{lebeda2025weightedfourierfactorizationsoptimal}, in terms of sum squared error, for every workload. We present experiments in Section \ref{sec:experiments} and conclusions in Section \ref{sec:conc}. Due to space constraints, proofs can be found in the \shortlong{full version of this paper \cite{qsarxiv}.}{appendix.}

\section{Notation}\label{sec:notation}
We use the following linear algebra notation.
Let $\kron$ be the Kronecker product; e.g., if the matrix 
$\mat{V}=\left(\begin{smallmatrix}
v_{0,0} & v_{0,1} & v_{0,2}\\
v_{1,0} & v_{1,1} & v_{1,2}\\
v_{2,0} & v_{2,1} & v_{2,2}\\
\end{smallmatrix}\right)$ and $\mat{W}$ is another matrix then $\mat{V}\kron\mat{W}$ 
can be written as the block matrix $\left(\begin{smallmatrix}
v_{0,0}\mat{W} & v_{0,1}\mat{W} & v_{0,2}\mat{W}\\
v_{1,0}\mat{W} & v_{1,1}\mat{W} & v_{1,2}\mat{W}\\
v_{2,0}\mat{W} & v_{2,1}\mat{W} & v_{2,2}\mat{W}\\
\end{smallmatrix}\right)$.
Let $\identity_d$ be the $d\times d$ identity matrix.
Let $\one_d$ be the d-dimensional column vector whose entries are all 1.

\paragraph*{\textbf{Data.}}
Let $\attr_1,\dots, \attr_\numattr$ be  attributes having discrete domains.
For $i=1,\dots,\numattr$, let $\attsize_i$ denote the domain size of attribute $\attr_i$ and let its domain
elements be $\attv^{(i)}_{0},\attv^{(i)}_{1}, \dots, \attv^{(i)}_{\attsize_i-1}$.
A value $\attv^{(i)}_j$ for attribute $\attr_i$ can also be represented in one-hot notation
as a $\attsize_i$-dimensional column vector $\hot^{(i)}_j$,  where the $j^\text{th}$ component is $1$
and all others are $0$. We index vectors using square brackets (e.g.,  $\hot^{(i)}_j[j]=1$).

A record $\rec$ on attributes $\attr_1,\dots, \attr_\numattr$ is a tuple that specifies a value
for each attribute. It can also be represented as a vector of dimension $\prod_{i=1}^\numattr \attsize_i$ by taking the Kronecker product $\kron$ of the one-hot encodings of its attribute values. For example, if the domain of $\attr_1$ is $\{\textit{low}, \textit{medium}, \textit{high}\}$, the domain of $\attr_2$ is $\{\textit{yes}, \textit{no}\}$, and the domain of $\attr_3$ is $\{0, 1, 2, 3, 4\}$, then the record $\rec=(\textit{medium}, \textit{yes}, 3)$
can be represented as  $\hot^{(1)}_1 \kron \hot^{(2)}_0 \kron \hot^{(3)}_3 = [0, 1, 0]^\top \kron [1, 0] ^\top\kron [0, 0, 0, 1, 0]^\top$.

A dataset $\data=\{\rec_1,\dots, \rec_\numrec\}$ is a set of records, where each record $\rec_i$ 
corresponds to one individual. In addition to this set view, a dataset $\data$ can also be viewed as a 
tensor/multidimensional array $\marginal$, in which the entry $\marginal[j_1, j_2,\dots, j_\numattr]$ is the
number of records for which $\attr_1=\attv^{(1)}_{j_1}, \attr_2=\attv^{(2)}_{j_2},\dots, \attr_\numattr=\attv^{(\numattr)}_{j_\numattr}$. A dataset $\data$ can also be represented as a \textbf{column} vector $\datavec$ which is the sum of the vector representations of the records in the dataset. There is a simple relationship between the tensor view $\marginal$ and vector view $\datavec$: if $\marginal$ is stored in row-major order (the default in numpy) then $\datavec=\marginal$.flatten() in Python.

\paragraph{\textbf{Marginals.}} 
Given a subset  $\attset=\{\attr_{i_1},\dots,\attr_{i_\ell}\}$ of attributes, a marginal on $\attset$ can also be represented in 3 different ways: as a set, as a tensor, and as a vector. The set view $\data_\attset$ is obtained from $\data$ by projecting on the attributes in $\attset$ (i.e., removing attributes not in $\attset$ from each record). If $\attset=\{\}$ is the empty set,  then $\data_{\{\}}=\{\perp, \perp, \dots, \perp\}$ and therefore the only information it retains is the total number of records. The tensor view $\marginal_{\attset}$ is obtained from $\data_{\attset}$ in the obvious way. That is, $\marginal_\attset$ is created from $\data_\attset$ in the same way as $\marginal$ is created from $\data$. Equivalently, if $\attset=\{\attr_{i_1},\dots,\attr_{i_\ell}\}$, then $\marginal_\attset$ is created from $\marginal$ by summing up over all dimensions that are not in $\{i_1, i_2,\dots, i_\ell\}$. In the special case of $\attset=\{\}$ then $\marginal_{\{\}}=\numrec$ (the number of records). Finally, the vector view $\datavec_\attset$ is obtained  by flattening $\marginal_\attset$. Once again, when $\attset=\{\}$ then $\datavec_{\{\}}=\numrec$.

\paragraph{\textbf{Linear Queries.}}
A linear query over a marginal can be represented as a multidimensional array $\query$ if we want to use it with the tensor view of a marginal. It can also be represented as a \textbf{row} vector $\queryv$ if we want to use it with the vector view of a marginal.  If $\marginal_\attset$ is a marginal with the same shape as $\query$ then the query answer is calculated by taking the element-wise product between $\query$ and $\marginal_\attset$ and then computing the sum (in the notation of numpy: $(\query * \marginal_\attset)$.sum()). In the vector view $\queryv$, we have $\queryv=\query$.flatten() and the query answer can be computed as the dot product $\queryv \cdot \datavec_\attset$. If all entries of $\query$ (or $\queryv$) are in $\{0, 1\}$, then we call it a \textbf{counting query}. A set of queries is called a \emph{workload} and is denoted by $\workload$. 

For convenience, our notation is summarized in Table \ref{tab:notation}.
\begin{table}[h]
\begin{center}
\caption{Table of Notation}\label{tab:notation}
\begin{tabular}{|lp{0.7\linewidth}|}\hline
$\data$: &Dataset of records $\{\rec_1,\dots, \rec_{\numrec}\}$.\\
$\numrec$: & number of records in $\data$\\
$\numattr$: & The number of attributes in the data schema.\\
$\attr_i$: & Attribute $i$ of the data schema. \\
$\attv^{(i)}_j$: & the $j^\text{th}$ domain element of attribute $\attr_i$.\\
$\attsize_i$: & Size of domain of Attribute $\attr_i$.\\
$\marginal$: & Tensor view of the data.\\
$\datavec$: & Vector view of the data (\textbf{column} vector)\\
$\attset$: & A subset of attributes.\\
$\data_\attset$: & Projection of $\data$ onto attributes in $\attset$.\\
$\marginal_\attset$: & Tensor view of $\data_\attset$.\\
$\datavec_\attset$: & Vector view of $\data_\attset$.\\
$\query$: & Linear query represented as a tensor.\\
$\queryv$: & Linear query represented as a \textbf{row} vector.\\
$\workload$: & A workload of queries.\\
$\mech$: & A privacy mechanism.\\
\hline
\end{tabular}
\end{center}
\end{table}
\section{Background}\label{sec:background}

\noindent\textbf{Differential Privacy.}\\
\noindent Differential privacy refers to a family of privacy definitions whose goal is to provide plausible deniability about the presence and contents of any record in a dataset. The most popular versions that support privacy accounting for Gaussian noise are: $(\epsilon,\delta)$-DP \cite{dworkKMM06:ourdata} (also known as approximate DP), Gaussian DP \cite{fdp}, zCDP \cite{zcdp}, and Renyi DP \cite{renyidp}. \sysname supports all of these definitions, but for simplicity, we focus on approximate DP and Gaussian DP.

The goal of these privacy definitions is to add noise into a mechanism's computation so that any two neighboring datasets $\data$ and $\data^\prime$ (i.e., datasets differing on the presence/absence of 1 record) become nearly indistinguishable. That is, if a mechanism $\mech$ is run on either $\data$ or $\data^\prime$, an attacker that observes the output of $\mech$ would have difficulty in determining which of the two databases was the input. Approximate and Gaussian DP only differ in the mathematical conditions they use to enforce indistinguishability:
\begin{definition}[Approximate DP \cite{dworkKMM06:ourdata}]
    Given privacy parameters $\epsilon\geq 0$ and $\delta\in[0,1]$, a mechanism $\mech$ satisfies $(\epsilon,\delta)$-DP if for all pairs of neighbors $\data,\data^\prime$ and all $S\subseteq\range(\mech)$, 
    $$P(\mech(\data)\in S)\leq e^\epsilon P(\mech(\data^\prime)\in S)+\delta$$
\end{definition}

\begin{definition}[Gaussian DP \cite{fdp}]
    Given privacy parameter $\mu\geq 0$, a mechanism $\mech$ satisfies $\mu$-Gaussian DP if for all pairs of neighbors $\data,\data^\prime$ and all $S\subseteq\range(\mech)$, 
    $$\Phi^{-1}(P(\mech(\data)\in S)\leq \Phi^{-1}(P(\mech(\data^\prime)\in S))+\mu$$
    where $\Phi$ is the CDF of the standard Gaussian.
\end{definition}

These privacy definitions can be satisfied by the Correlated Gaussian Linear Mechanism, defined as follows.
\begin{definition}[Correlated Gaussian Linear Mechanism, Privacy Cost \cite{commonmech}]\label{def:glm}
    Given a dataset represented as a vector $\datavec$ of counts, a Correlated Gaussian Linear Mechanism has the form $\mech(\datavec)=\bmat \datavec + N(\zero, \covar)$, where $\bmat$ is a  matrix referred to as the \textbf{strategy matrix} and $\covar$ controls the covariance of the Gaussian noise added to the vector of query answers $\bmat\datavec$. The \textbf{privacy cost} of $\mech$ is defined as $\max_i (\bmat^\top \covar^{-1}\bmat)[i,i]$ (i.e., largest diagonal entry of $\bmat^\top \covar^{-1}\bmat$).
\end{definition}

Computing the approximate DP or Gaussian DP privacy parameters for a collection of correlated Gaussian linear mechanisms is done using the following result:
\begin{theorem}[\cite{commonmech}]\label{thm:overall}
    Let $\mech_1,\dots, \mech_k$ be a collection of correlated Gaussian linear mechanisms with respective strategy matrices $\bmat_1,\dots, \bmat_k$ and noise covariance matrices $\covar_1,\dots, \covar_k$. Define the overall privacy cost as $\pcost(\mech_1,\dots,\mech_k)=\max_i \left(\sum_{j=1}^k \bmat^\top_j \covar_j^{-1}\bmat_j\right)[i,i]$. If $\pcost(\mech_1,\dots,\mech_k)\leq \pcostbound$ then:
    \begin{itemize}[leftmargin=*]
        \item Releasing the outputs of $\mech_1,\dots, \mech_k$ collectively satisfies $(\epsilon,\delta)$-DP with $\delta=\Phi(\sqrt{\pcostbound}/2-\epsilon/\sqrt{\pcostbound})-e^\epsilon \Phi(-\sqrt{\pcostbound}/2-\epsilon/\sqrt{\pcostbound})$.
        \item Releasing the outputs of $\mech_1,\dots, \mech_k$ collectively satisfies $\mu$-Gaussian DP with $\mu=\sqrt{\pcostbound}$.
    \end{itemize}
\end{theorem}

An important feature of differential privacy and its variations is postprocessing invariance. If $\mech$ satisfies $(\epsilon,\delta)$-DP (resp., $\mu$-Gaussian DP) and $\randalg$ is a possibly randomized algorithm whose domain contains the range of $\mech$, then the algorithm $\randalg\circ\mech$ (which outputs $\randalg(\mech(\data))$) satisfies $(\epsilon,\delta)$-DP (resp., $\mu$-Gaussian DP) with the same values of the privacy parameters.

\vspace{0.5em}\noindent\textbf{Matrix Mechanisms.}\\
The matrix mechanism framework \cite{LMHMR15,edmonds2020power} uses the correlated Gaussian linear mechanism\footnote{There also exist matrix mechanisms that use Laplace noise \cite{LMHMR15,YZWXYH12}.} to answer a workload $\workload$ of linear queries. The workload is typically represented as a matrix $\mat{W}$, whose rows are the vector representations of queries in $\workload$. The weight of each query $\queryv$, denoted as $\weight(\queryv)$, indicates its importance. The \emph{selection} step chooses matrices $\mat{L}$, $\bmat$, and covariance matrix $\covar$ with the condition that $\mat{W}=\mat{L}\bmat$. The \emph{measure} step uses the correlated Gaussian linear mechanism determined by the strategy matrix $\bmat$ and covariance matrix $\covar$ to get noisy
answers $\vec{z}$ to the strategy queries: $\vec{z}=\bmat\datavec + N(\zero, \covar)$. The reconstruction phase computes $\mat{L}\vec{z}$, which are noisy answers to the original workload queries $\mat{W}$. These noisy answers have the following properties:
\begin{itemize}[leftmargin=*]
    \item $E[L\vec{z}]=\mat{W}\datavec$, so that the reconstructed answers are unbiased.
    \item Covar$[L\vec{z}]=\mat{L}\covar\mat{L}^\top$. Hence the covariance between the $i\text{th}$ query and $j^\text{th}$ query is the $(i,j)$ entry in this matrix and the $(i,i)$ entry is the variance of the $i^\text{th}$ query.
    \item The privacy cost is: $\max_i (\bmat^\top\covar^{-1}\bmat)[i,i]$.
\end{itemize}
Hence to minimize the weighted sum squared error of the reconstructed query answers, the selection phase should set $\mat{L}=\mat{W}\bmat^\dagger$ (the pseudoinverse of $\bmat$) \cite{LMHMR15} and then choose $\bmat$ and $\covar$ to minimize the weighted sum of the diagonal entries of $\mat{L}\covar\mat{L}^\top$ while (1) keeping the privacy $\leq$ some user-defined threshold $\pcostbound$ and (2) ensuring the rowspace of $\bmat$ contains the rows of $\mat{W}$. When $\mat{W}^\top\mat{W}$ fits in memory, these requirements can be satisfied by solving a semidefinite programming problem using the algorithm of Yuan et al. \cite{YYZH16} or the accelerated variation used within HDMM \cite{mckenna2018optimizing}.


\section{Related Work}\label{sec:related}
There is a rich and ever-growing literature on differentially private query-answering.
Some of the most successful approaches \cite{McKenna_Miklau_Sheldon_2021} combine ideas
from the MWEM \cite{hardt2012simple} and matrix mechanism \cite{LMHMR15} line of work.
MWEM-style algorithms \cite{hardt2012simple,gaboardi2014dual,aydore2021differentially,zhang2018ektelo,liu2021leveraging,aim,mullins2024efficient,fuentes2026fast}
maintain a privacy-preserving data synopsis and iteratively identify which queries are poorly answered by it. They obtain noisy answers to those queries (or to alternative strategy queries \cite{mullins2024efficient,fuentes2026fast}) and then update the data synopsis. These query answers typically have bias and data-dependent noise, but tend to be very accurate.

The matrix mechanism line of work \cite{privelet,hay2009boosting,LMHMR15,YZWXYH12,YYZH16,nikolov2014new,edmonds2020power,hbtree,xiao2020optimizing,mckenna2018optimizing,rplanner,rplus,lebeda2025weightedfourierfactorizationsoptimal} answers workload queries by first finding surrogate queries, known as strategy queries, that can be answered with small amounts of noise. Workload query answers are then reconstructed from noisy answers to the strategy queries. They are unbiased and have data-independent variances. Matrix mechanisms are often used inside of other differentially private query-answering systems  \cite{tdahdsr,aim,privbayes,hbtree,fuentes2026fast,mullins2024efficient,mckenna2025scaling}.
Although some matrix mechanisms support pure $\epsilon$-differential privacy and use Laplace noise \cite{LMHMR15,YZWXYH12}, better privacy utility trade-offs are obtained with relaxed variants of differential privacy \cite{dworkKMM06:ourdata,zcdp,renyidp,fdp} and Gaussian noise.

Most work on matrix mechanisms tries to optimize for the weighted sum squared error of workload queries, although optimizing a weighted max \cite{nikolov2014new,edmonds2020power,xiao2020optimizing} is also important in practice. Yuan et al.  \cite{YYZH16} provide an optimal algorithm for optimizing the weighted sum-squared error in low dimensional settings, while Xiao et al. \cite{xiao2020optimizing} provide a corresponding algorithm for the weighted max objective.
RP \cite{rplanner} and  WFF \cite{lebeda2025weightedfourierfactorizationsoptimal} can support
a much wider range of convex objectives.

Workloads that can be answered from a collection of data marginals, such as marginal and range queries, are important in practice \cite{tdahdsr} and HDMM was the first scalable matrix mechanism that can optimize for them. However, its runtime and memory requirements were limited by the size of the data domain.
RP \cite{rplanner} overcame this limitation by adding noise to carefully chosen strategy queries called the residual basis. The scale of the noise added to the residual basis queries can be tuned to optimize workload accuracy. RP is optimal for marginal queries and scales with the size of the workload rather than the exponentially sized data domain. RP+ \cite{rplus} is an extension of RP that allows a user to customize which basis should be used, but it does not provide a way of choosing a good basis.
WFF \cite{lebeda2025weightedfourierfactorizationsoptimal} proposed an alternative basis that is optimal for marginals and circular product queries (e.g., circular range queries). It computes the multidimensional Fourier transform of data marginals and optimizes for workload accuracy by tuning the variance of the noise added to the Fourier coefficients. Our proposed mechanism \sysname maintains the scalability of \cite{rplanner,rplus,lebeda2025weightedfourierfactorizationsoptimal}, while  provably dominating them under sum-squared-error on any workload.


\section{Overview of {\protect\sysname}}\label{sec:overview}

We first provide an overview of, and  intuition behind, \sysname, and then present the full algorithm in Section \ref{sec:decomp}. Afterwards, we show how prior work like RP, RP+, and WFF can be recast as special cases of \sysname, and consequently \sysname dominates them on all workloads under sum-squared error.

Let $\datavec$ be the vector representation of the data, let  $\attset\subseteq\{\attr_1,\dots, \attr_\numattr\}$ be a subset of attributes,  and let $\datavec_\attset$ be the vector representation of the marginal on $\attset$ (as in Section \ref{sec:notation}). The vector $\datavec_\attset$ can be computed from $\datavec$ by the matrix multiplication: $\mquery{\attset}\datavec = \datavec_\attset$, where
\begin{align}
    \mquery{\attset} &= \mat{V}_1 \kron \mat{V}_2\kron \cdots \kron \mat{V}_\numattr\nonumber\\
    \text{where } \mat{V}_i &=
    \begin{cases}
        \identity_{\attsize_i} & \text{ if }\attr_i\in\attset\\
        \one^T_{\attsize_i} &\text{ otherwise }
    \end{cases}\nonumber
\end{align}
and so a linear query $\queryv$ over $\datavec_{\attset}$ (whose true answer is $\queryv\cdot\datavec_{\attset}$) can be turned into a linear query $\queryv^\prime$ over $\datavec$ by setting $\queryv^\prime=\queryv\mquery{\attset}$.
Given two subsets of attributes $\attset$ and $\attset^\prime$, the matrices $\mquery{\attset}$ and $\mquery{\attset^\prime}$ have linear dependencies between them. So, Xiao et al.  proposed the
concept of a residual space $\resid_{\attset}$ on $\attset$ \cite{rplanner}. $\resid_{\attset}$ is a linear subspace
of the row space of $\mquery{\attset}$ and any two residual spaces $\resid_{\attset}$ and $\resid_{\attset^\prime}$ (for $\attset\neq\attset^\prime$) are mutually orthogonal. Intuitively, $\resid_{\attset}$ corresponds to information that is unique to the marginal on $\attset$ \cite{rplanner}. Formally, $\resid_\attset$ is defined as:
\begin{align}
    \resid_{\attset} &= \text{rowspan}\left(\mat{W}_1 \kron \mat{W}_2\kron \cdots \kron \mat{W}_\numattr\right)\\
    \text{where } \mat{W}_i &= 
    \begin{cases}
        \identity_{\attsize_i} - \frac{1}{\attsize_i}\one_{\attsize_i}\one^T_{\attsize_i}& \text{ if }\attr_i\in\attset\\
        \one^T_{\attsize_i} &\text{ otherwise }
    \end{cases}\nonumber
\end{align}
RP \cite{rplanner} and WFF \cite{lebeda2025weightedfourierfactorizationsoptimal} proposed specific
linearly independent bases for each $\resid_\attset$ and used those bases as the rows in the strategy matrix
$\bmat$ of their Gaussian linear mechanisms (Definition \ref{def:glm}). In the case of RP+ \cite{rplus}, the user
was asked to supply the bases. \sysname takes a different approach: given a workload $\workload$ of linear queries that can be answered from a collection of 
marginals, \sysname breaks apart each $\queryv\in \workload$ by projecting it 
onto the different residual spaces. Subqueries over the same residual space are grouped together into a subworkload.
This results in a collection of low-dimensional subworkloads that are small enough to be independently optimized by 
an exact non-scalable algorithm \cite{yuan2015optimizing}. The matrix mechanisms created for each subworkload are stitched
together to create a matrix mechanism for the original workload.
Specifically, the steps are:
\begin{enumerate}[leftmargin=*]
    \item First,  queries in $\workload$ are organized by the marginals that can answer them. For example, a query that uses only attributes $\attr_2$ and $\attr_8$ is answerable from the marginal $\datavec_\attset$ on $\attset={\{\attr_2,\attr_8\}}$.
    \item Each $\queryv\in\workload$ that can be answered by $\datavec_\attset$ is projected onto the residual spaces $\resid_{\attset^\prime}$ for $\attset^\prime\subseteq\attset$.  For example, if query $\queryv$ can be answered from marginal $\datavec_{\{\attr_2,\attr_8\}}$, it will be decomposed into 4 pieces that we denote as $\queryv_{\deco\{\}}, \queryv_{\deco\{\attr_2\}}, \queryv_{\deco\{\attr_8\}},$ and  $\queryv_{\deco\{\attr_2, \attr_8\}}$. Subquery $\queryv_{\deco\{\}}$ can be answered from $\datavec_{\{\}}$, subquery $\queryv_{\deco\{\attr_2\}}$ can be answered from $\datavec_{\{\attr_2\}}$, etc. The answer to the original query $\queryv$ can computed by adding up the answers to the subqueries. 
    The details are presented in Section \ref{sec:decomp}.
    \item After this decomposition, \sysname creates subworkloads by grouping together subqueries that belong to the same residual space. That is,  subworkload  $\workload_{\deco\attset}$ consists of all subqueries that result from projections onto the residual space $\resid_{\attset}$. For each subworkload $\workload_{\deco\attset}$, \sysname can use any matrix mechanism construction technique, including the optimal low-dimensional algorithm from
    \cite{YYZH16}, to create a matrix mechanism $\mech_{\deco\attset}$ that minimizes the sum-squared error over the subworkload at a privacy cost of 1. 
    %
    \item Given the mechanisms $\mech_{\deco\attset}$ for different sets $\attset$, \sysname rescales the noise that they use, in order to optimize for the global weighted sum squared error of queries in $\workload$ subject to a given bound $\pcostbound$ on the privacy cost. Then \sysname runs those mechanisms to obtain noisy answers for the subqueries in each subworkload. 
    \item To reconstruct an answer to a query $\queryv\in\workload$, $\sysname$ simply adds up the noisy answers to its subqueries.
\end{enumerate}

\section{The {\protect \sysname} Algorithm}\label{sec:decomp}
We  present the pseudocode for \sysname in the next subsection and then explain the key details
in the following subsections.
\subsection{The Overall Algorithm}\label{sec:overallalg}

\begin{algorithm}[h]
   \DontPrintSemicolon
   \KwIn{$\data$,  $\workload$, weight for each $\queryv\in\workload$, target privacy cost $\pcostbound$}
   \For{$\queryv\in\workload$\label{line:part1start}}{
       $\attset\gets$ attributes used by $\queryv$.\;
       \For(\tcp*[f]{Decompose $\queryv$}){$\attset^\prime\subseteq\attset$}{
          Create $\queryv_{\deco\attset^\prime}$ using the method  in Section \ref{sec:vectordeco}.\;
          $\weight(\queryv_{\deco\attset^\prime})\gets \weight(\queryv)$.\;
          Add $\queryv_{\deco\attset^\prime}$ to the subworkload $\workload_{\deco\attset^\prime}$.\label{line:part1end}
       }
   }
   \tcp{Optimize subworkloads independently.}
   \For{each non-empty subworkload $\workload_{\deco\attset^\prime}$\label{line:part2start}}{
       Optimize $\workload_{\deco\attset^\prime}$ using a low-dimensional solver.\;
       $\mech_{\attset^\prime}$ is the Gaussian Linear Mechanism provided by the solver. I.e., $\mech_{\attset^\prime}(\datavec_{\attset^\prime})\equiv\bmat_{\attset^\prime} \datavec_{\attset^\prime}+ N(\mat{0},\covar_{\attset^\prime})$.\;
       $\recon_{\attset^\prime}$ is a linear reconstruction function provided by the solver. It can take a query and the output of $\mech_{\attset^\prime}$ and return an unbiased estimate of the query answer.\;
       $\varfun_{\attset^\prime}$ is the variance function provided by the solver. It takes a query and returns that variance of the reconstructed query answer given by $\recon_{\attset^\prime}$.\;
       $\loss_{\attset^\prime}$ is the weighted sum square variance of queries in $\workload_{\deco\attset^\prime}$.\label{line:part2end}
   }
   \tcp{Assemble $\&$ run  overall mechanism for $\workload$.}
   $\gamma\gets \frac{\sum_{\attset^\prime} \sqrt{\loss_{\attset^\prime}}}{\pcostbound}$\label{line:rescale1}\;
   \For{each non-empty subworkload $\workload_{\deco\attset^\prime}$}{
       $\sigma^2_{\attset^\prime}= \gamma/\sqrt{\loss_{\attset^\prime}}$.\label{line:rescale2}\;
       Define $\mech^*_{\attset^\prime}(\datavec_{\attset^\prime})\equiv\bmat_{\attset^\prime} \datavec_{\attset^\prime}+ N(\mat{0},\sigma^2_{\attset^\prime}\covar_{\attset^\prime})$\label{line:rescale3}\;
       $\outp_{\attset^\prime} \gets \mech^*_{\attset^\prime}(\datavec_{\attset^\prime})$\tcp{run the mechanism}\label{line:measure}
    }
    \tcp{Reconstruction phase (query answering)}
    \For{$\queryv\in\workload$\label{line:part4start}}{
        answer$(\queryv)\gets 0$ $\quad$and$\quad$ 
        variance$(\queryv)\gets 0$\;
        $\attset\gets$ attributes used by $\queryv$.\;
        \For{$\attset^\prime\subseteq\attset$}{
           answer$(\queryv) \text{ += } \recon_{\attset^\prime}(\queryv_{\deco\attset^\prime},\outp_{\attset^\prime})$\;
           variance$(\queryv)  \text{ += } \varfun_{\attset^\prime}(\queryv_{\deco\attset^\prime})$\label{line:part4end}
        }
    }
    \Return reconstructed answer to each query and its variance
\caption{Pseudocode for \sysname}\label{alg:smasher}
\end{algorithm}


The pseudocode for \sysname is shown in Algorithm \ref{alg:smasher}.
Every $\queryv\in\workload$ has a weight, denoted as $\weight(\queryv)$.
The goal is to provide unbiased noisy answers to these queries while minimizing the weighted sum of their variances:
\begin{align}
    \loss &= \sum_{\queryv\in\workload} \weight(\queryv)*\text{Variance}(\text{noisy answer to }\queryv)\label{eqn:obj}
\end{align}
while keeping the privacy cost bounded by $\pcostbound$.

The first for-loop, covering Lines \ref{line:part1start} through \ref{line:part1end}, decomposes each query $\queryv$ into its subqueries by projecting it onto the appropriate residual spaces. If a query accesses a subset $\attset$ of attributes (i.e., $\queryv$ is answerable from $\datavec_\attset$), then there is one subquery for each $\attset^\prime\subseteq \attset$. Thus, for example, a 2-d range query on income and education level for a specific gender accesses 3 attributes (income, education level, gender) and would be decomposed into 8 subqueries. Then each subquery is assigned to the appropriate subworkload. The weight of each subquery is equal to the weight of the original query. The formula for the decomposition is in Section \ref{sec:vectordeco}.

The next for-loop, covering Lines \ref{line:part2start} through \ref{line:part2end} iterates through the subworkloads. For each subworkload $\workload_{\deco\attset^\prime}$, it uses a low-dimensional algorithm (e.g.,  \cite{yuan2015optimizing}) to create a  matrix mechanism for the queries in $\workload_{\deco\attset^\prime}$ that minimizes their weighted sum-squared error at a privacy cost of 1. Each subworkload can be processed independently and in parallel. Note: data are \emph{not} accessed in this step. The details are presented in Section \ref{sec:solvers}.  Also note that algorithms equivalent to RP, RP+, and WFF can be recovered as special cases by using weaker matrix mechanisms than \cite{yuan2015optimizing} in this step (we explain this in Section \ref{sec:accuracy}). 
%

A matrix mechanism created in this step consists of the following components (1) the gaussian linear mechanism $\mech_{\attset^\prime}(\datavec_{\attset^\prime})=\bmat_{\attset^\prime} \datavec_{\attset^\prime}+ N(\mat{0},\covar_{\attset^\prime})$, (2) a linear reconstruction function $\recon_{\attset^\prime}$ that can answer a 
query  $\query_{\deco\attset^\prime}\in \workload_{\deco\attset^\prime}$ if one has the noisy output of $\mech_{\attset^\prime}(\datavec_{\attset^\prime})$, and (3) a function $\varfun_{\attset^\prime}$ such that $\varfun_{\attset^\prime}(\query_{\deco\attset^\prime})$ is the variance of the reconstructed answer to $\query_{\deco\attset^\prime}$. Neither  $\recon_{\attset^\prime}$ nor $\varfun_{\attset^\prime}$ require access to data, so that one can calculate the variance of any query before seeing the data.


Next, the amount of noise used by each $\mech_{\attset^\prime}$ needs to be rescaled so that the total privacy cost is  $\leq\pcostbound$. Lines \ref{line:rescale1} and \ref{line:rescale2} compute the optimal rescaling (Section \ref{sec:rescale}) and Line \ref{line:rescale3} uses it to adjust each mechanism's noise. Line \ref{line:measure}
then runs the mechanisms (this is the only line  that accesses data).
%

Lines \ref{line:part4start} through \ref{line:part4end} reconstruct the answer to every $\queryv\in\workload$ and compute its variance by reconstructing the answers to the subqueries of $\queryv$ (along with their variances) and adding them together.

\subsection{Query Decomposition Details}\label{sec:vectordeco}
We next explain how \sysname decomposes each $\queryv$.
Suppose a query $\queryv$ uses the attribute set $\attset=\{\attr_{i_1},\dots,\attr_{i_\ell}\}$.  For each $\attset^\prime\subseteq \attset$, define the subquery $\queryv_{\deco\attset^\prime}$ as:
\begin{align}
    \queryv_{\deco\attset^\prime} &= \queryv \left(\mat{W}_{i_1} \kron \mat{W}_{i_2}\kron\cdots\kron \mat{W}_{i_\ell}\right)\label{eq:vecdecomp}\\
    &\text{where }
    \mat{W}_{i_j} =
    \begin{cases}
        \frac{1}{\attsize_{i_j}}\one_{\attsize_{i_j}}  & \text{ if }\attsize_{i_j}\notin \attset^\prime\\
        \identity_{\attsize_{i_j}} - \frac{1}{\attsize_{i_j}}\one_{\attsize_{i_j}}\one_{\attsize_{i_j}}^T &\text{otherwise}
    \end{cases}\nonumber
\end{align}

\begin{figure}[t]
    \centering
    $\query:$
    \begin{tabular}{c|c|c|c|}
      \multicolumn{1}{c}{}   & \multicolumn{1}{c}{a} &  \multicolumn{1}{c}{b} & \multicolumn{1}{c}{c}\\\cline{2-4}
       yes  & 0 & 1 & 1\\\cline{2-4}
       no  & 0 & 0 & 1\\\cline{2-4}
    \end{tabular}
    $\qquad\queryv: [0, 1, 1, 0, 0, 1]$
    \caption{A query answerable by the marginal on $\attset={\{\attr_1,\attr_2\}}$
    with domain($\attr_1$)=$\{$yes, no$\}$ and domain($\attr_2$)=$\{$a,b,c$\}$. The tensor form $\query$ (left) and
    vector form $\queryv$ (right) are shown. The query asks for the number of records that satisfy ($\attr_1=$ yes,
    $\attr_2=$ b) $\vee$ ($\attr_1=$ yes, $\attr_2=$ c) $\vee$ ($\attr_1= $ no, $\attr_2=$ c).}
    \label{fig:decomposition_pre}
\end{figure}

\begin{figure}[t]
    \centering
    \begin{tabular}{rr}
    $\query_{\deco\{\}}:$  &  \dotfill $~$1/2  \\
      $\queryv_{\deco\{\}}$: & \dotfill $~$1/2\\
        $\query_{\deco\{\attr_1\}}:$ & \dotfill
    \begin{tabular}{c|c|}\cline{2-2}
       yes  & 1/6 \\\cline{2-2}
       no  & -1/6\\\cline{2-2}
    \end{tabular} \\
     $\queryv_{\deco\{\attr_1\}}$:  & \dotfill $[1/6,~ -1/6]$\\
       $\query_{\deco\{\attr_2\}}:$ &\dotfill
    \begin{tabular}{|c|c|c|}
       \multicolumn{1}{c}{a} &  \multicolumn{1}{c}{b} & \multicolumn{1}{c}{c}\\\cline{1-3}
         -1/2 & 0 & 1/2\\\cline{1-3}
    \end{tabular} \\
     $\queryv_{\deco\{\attr_2\}}$: & \dotfill $[-1/2, ~0, ~1/2]$\\
    $\query_{\deco\{\attr_1, \attr_2\}}:$ & \dotfill
    \begin{tabular}{c|c|c|c|}
      \multicolumn{1}{c}{}   & \multicolumn{1}{c}{a} &  \multicolumn{1}{c}{b} & \multicolumn{1}{c}{c}\\\cline{2-4}
       yes  & $-1/6$ & $1/3$ & $-1/6$\\\cline{2-4}
       no  & $1/6$ & $-1/3$ & $1/6$\\\cline{2-4}
    \end{tabular} \\
     $\queryv_{\deco\{\attr_1, \attr_2\}}$:&  \dotfill $[-1/6,~ 1/3,~ -1/6,~ 1/6,~ -1/3,~ 1/6]$\\
   \end{tabular}
    \caption{The decomposition of $\query$ from Figure \ref{fig:decomposition_pre} into the subqueries $\query_{\deco\{\}}, \query_{\deco\{\attr_1\}}, \query_{\deco\{\attr_2\}},$ and  $\query_{\deco\{\attr_1, \attr_2\}}$}
    \label{fig:decomposition_post}
\end{figure}


For example, consider the query whose vector representation is shown in Figure \ref{fig:decomposition_pre} (right). The dataset can have many attributes, but the query can be answered from $\datavec_{\{\attr_1,\attr_2\}}$, the vector representation of the marginal on  $\attset=\{\attr_1,\attr_2\}$. Attribute $\attr_1$ has domain size $\attsize_1=2$ and $\attr_2$ has domain size $\attsize_2=3$. This query is decomposed into 4 subqueries: $\queryv_{\deco\{\}}, \queryv_{\deco\{\attr_1\}}, \queryv_{\deco\{\attr_2\}}, $ and $\queryv_{\deco\{\attr_1,\attr_2\}}$.

To calculate  $\queryv_{\deco\{\}}$ using Equation \ref{eq:vecdecomp}, we multiply $\queryv$ on the right by $(\frac{1}{2}\one_2) \kron (\frac{1}{3}\one_3)$, which is a 6-dimensional vector whose entries are all $1/6$. Thus $\queryv_{\deco\{\}}$ is a scalar, the average of the entries in $\queryv$. 

To calculate $\queryv_{\deco(\attr_2)}$, we multiply $\queryv$ on the right by
$(\frac{1}{2}\one_2) \kron (\identity_3 - \frac{1}{3}\one_3\one_3^T)$, i.e.,
$\left[\begin{smallmatrix}1/2\\1/2\end{smallmatrix}\right]\kron \left[\begin{smallmatrix}\phantom{-}2/3 &-1/3 & -1/3\\-1/3 & \phantom{-}2/3 &-1/3\\-1/3 & -1/3 & \phantom{-}2/3\end{smallmatrix}\right]$, and so on.
The four queries that $\queryv$ is decomposed into are shown in Figure \ref{fig:decomposition_post}. Decomposed queries have the following important properties:
\begin{enumerate}[leftmargin=*]
    \item Each subquery $\queryv_{\deco\attset^\prime}$ is answerable from the marginal $\datavec_{\attset^\prime}$  by taking dot products: $\queryv_{\deco\attset^\prime}\cdot\datavec_{\attset^\prime}$.
    \item If $\attset^\prime\neq \attset^{\prime\prime}$ then the decompositions into  $\queryv_{\deco\attset^\prime}$ and  $\queryv_{\deco\attset^{\prime\prime}}$ are orthogonal. That  is, if $\queryv_{\deco\attset^\prime}=\queryv\mat{Q}_1$ and  $\queryv_{\deco\attset^{\prime\prime}}=\queryv\mat{Q}_2$ then the columns of $\mat{Q}_1$ are orthogonal to the columns of $\mat{Q}_2$.
    \item The answer to $\queryv$ is equal to the sum of the answers of the subqueries. That is $\queryv\cdot\datavec_\attset=\sum\limits_{\attset^\prime\subseteq\attset} \queryv_{\deco\attset^\prime}\cdot\datavec_{\attset^\prime}$.
    \item If the subqueries are answered by matrix mechanisms that are independent of each other, the variance of $\queryv$ equals the sum of the variances of its subqueries.
\end{enumerate}
The first and fourth properties are obvious and the other two are proved in the following Theorem \ref{thm:decomp}. The second property means that the information contained in one subquery $\queryv_{\deco\attset^\prime}$ is independent of information contained in a different subquery $\queryv_{\deco\attset^{\prime\prime}}$. That is, if both are rewritten as query vectors over the full data vector $\datavec$, then those query vectors are orthogonal.
Since $\queryv_{\deco\attset^\prime}$ is later placed into subworkload $\workload_{\deco\attset^\prime}$ and $\queryv_{\deco\attset^{\prime\prime}}$ is placed into $\workload_{\deco\attset^{\prime\prime}}$, this means that subworkloads are orthogonal and linearly independent (i.e., if their subqueries are rewritten as query vectors over the full data vector $\datavec$, the subqueries in $\workload_{\deco\attset^{\prime}}$ are orthogonal to subqueries in $\workload_{\deco\attset^{\prime\prime}}$). Hence, noisy answers to subqueries in $\workload_{\deco\attset^{\prime}}$ carry no information about noisy answers to subqueries in $\workload_{\deco\attset^{\prime\prime}}$ and hence subworkloads can be optimized independently of each other.

The last two properties in the list above are the reason that reconstruction (Lines \ref{line:part4start} through \ref{line:part4end} in Algorithm \ref{alg:smasher}) is correct, since subqueries are placed in different subworkloads, and the subworkloads are optimized separately.

\begin{theoremE}[][category=decomp,proof here]\label{thm:decomp}
Let $\attset=\{\attr_{i_1},\dots,\attr_{i_\ell}\}$ be a subset of attributes. Let $\queryv$ be a query that is answerable from the marginal $\datavec_{\attset}$. For any $\attset^\prime\subseteq\attset$, define $\mat{Q}_{\attset^\prime}$ using Equation \ref{eq:vecdecomp} so that $\queryv_{\deco\attset^\prime}=\queryv\mat{Q}_{\attset^\prime}$. Then
$\queryv \cdot \datavec_{\attset} = \sum\limits_{\attset^\prime\subseteq \attset} \queryv_{\deco\attset^\prime}\cdot \datavec_{\attset^\prime}$.
Furthermore, if $\attset^\prime\neq\attset^{\prime\prime}$ then $\mat{Q}_{\attset^\prime}^\top\mat{Q}_{\attset^{\prime\prime}}=\mat{0}$.
\end{theoremE}
\begin{proofE}
    We prove orthogonality first. Represent the matrices in their kronecker product forms: $\mat{Q}_{\attset^\prime}=\mat{W}^\prime_1\kron \mat{W}^\prime_2\kron\cdots\kron\mat{W}^\prime_\ell$ and $\mat{Q}_{\attset^{\prime\prime}}=\mat{W}^{\prime\prime}_1\kron \mat{W}^{\prime\prime}_2\kron\cdots\kron\mat{W}^{\prime\prime}_\ell$. Then since $\attset^\prime\neq \attset^{\prime\prime}$, there exists an index $j$ such that $\mat{W}^\prime_j\neq\mat{W}^{\prime\prime}_j$. Specifically, one of those matrices must be equal to $\frac{1}{\attsize_j}\one_{\attsize_{i_j}}$ and the other must be equal to $\identity_{\attsize_{i_j}}-\frac{1}{\attsize_j}\one_{\attsize_{i_j}}\one_{\attsize_{i_j}}^T$. In this case, clearly $(\mat{W}_j^\prime)^T\mat{W}_j^{\prime\prime}=\mat{0}$ and $(\mat{W}_j^{\prime\prime})^T\mat{W}_j^{\prime}=\mat{0}$. Thus, 
\begin{multline*}
\mat{Q}_{\attset^\prime}^T\mat{Q}_{\attset^{\prime\prime}} = ((\mat{W}_1^\prime)^T\mat{W}_1^{\prime\prime}) \kron \cdots \\
\kron((\mat{W}_j^\prime)^T\mat{W}_j^{\prime\prime})\kron \cdots\kron ((\mat{W}_\ell^\prime)^T\mat{W}_\ell^{\prime\prime})=\mat{0}
\end{multline*}

    Next, we note that for any $\attset^\prime\subseteq\attset$, then $\datavec_{\attset^\prime}$ can be obtained from $\datavec_\attset$ by left multiplication by a matrix. Specifically,
    \begin{align*}
    \datavec_{\attset^\prime} &= \left(\mat{V}^\prime_1\kron\mat{V}^\prime_2\kron\cdots\kron\mat{V}^\prime_\ell\right)\datavec_\attset\\
    &\text{where }\mat{V}_j^\prime=
    \begin{cases}
        \one_{\attsize_{i_j}}^T &\text{ if }\attr_i\notin\attset^\prime\\
        \identity_{\attsize_{i_j}} &\text{ otherwise}
    \end{cases}
    \end{align*}
    Let $\mat{Q}_{\attset^\prime}=\mat{W}_1^\prime\kron\cdots\kron\mat{W}_\ell^\prime$ be the matrix that turns $\queryv$ into $\queryv_{\deco\attset^\prime}$ by right multiplication as in Equation \ref{eq:vecdecomp} and
    let $\mat{V}_{\attset^\prime}=\mat{V}_1^\prime\kron\cdots\kron\mat{V}_\ell^\prime$ be the matrix that converts
    $\datavec_{\attset}$ into $\datavec_{\attset^\prime}$ by left multiplication. It is clear that
    $\mat{W}_j^\prime=\frac{1}{\attsize_{i_j}}\one_{\attsize_{i_j}}$ if and only if $\mat{V}^\prime_j=\one_{\attsize_{i_j}}^T$. Therefore 
    \begin{align*}
        \mat{W}^\prime_j\mat{V}^\prime_j &=
        \begin{cases}
            \frac{1}{\attsize_{i_j}}\one_{\attsize_{i_j}}\one_{\attsize_{i_j}}^T &\text{ if }\attr_{i_j}\notin \attset^\prime\\
            \identity_{\attsize_{i_j}} - \frac{1}{\attsize_{i_j}}\one_{\attsize_{i_j}}\one_{\attsize_{i_j}}^T &\text{ otherwise}
        \end{cases}
    \end{align*}
    Therefore
    \begin{align*}
        &\sum_{\attset^\prime\subseteq \attset} \queryv_{\deco\attset^\prime}\cdot \datavec_{\attset^\prime}\\
        &=\sum_{\attset^\prime\subseteq \attset} \queryv \mat{Q}_{\attset^\prime} \mat{V}_{\attset^\prime}\datavec_{\attset}
        =\queryv \left(\sum_{\attset^\prime\subseteq \attset} \mat{Q}_{\attset^\prime} \mat{V}_{\attset^\prime}\right)\datavec_{\attset}\\
        & =\queryv \left(\sum_{\attset^\prime\subseteq \attset}\Kron_{j=1}^\ell 
        \left(\substack{
            \frac{1}{\attsize_{i_j}}\one_{\attsize_{i_j}}\one_{\attsize_{i_j}}^T \qquad\text{ if }\attr_{i_j}\notin \attset^\prime\\
            \identity_{\attsize_{i_j}} - \frac{1}{\attsize_{i_j}}\one_{\attsize_{i_j}}\one_{\attsize_{i_j}}^T \text{ otherwise}
        }\right)
        \right)\datavec_{\attset}\\
        &=\queryv\!\left(\Kron_{j=1}^\ell\!\left(\!\left(\identity_{\attsize_{i_j}} - \tfrac{1}{\attsize_{i_j}}\one_{\attsize_{i_j}}\one_{\attsize_{i_j}}^T\right)\right.\right.\\
        &\qquad\qquad\left.\left.{}+ \tfrac{1}{\attsize_{i_j}}\one_{\attsize_{i_j}}\one_{\attsize_{i_j}}^T\right)\!\right) \datavec_{\attset}
        \\
        &=\queryv\datavec_{\attset}
    \end{align*}
    where the second-to-last equality is obtained by transforming the expression
    $\sum_{\attset^\prime\subseteq\attset} f_0(\attset^\prime)$
    (here $f_0$ is that Kronecker product on the 3rd equation line) into
    $\sum_{\attset^\prime\subseteq \attset\setminus\{\attr_{i_1}\}} f_1(\attset^\prime)$,
    where $f_1(\attset^\prime)=f_0(\attset^\prime) + f_0(\attset^\prime\cup\{\attr_{i_1}\})$,
    then converting that expression into
    $\sum_{\attset^\prime\subseteq \attset\setminus\{\attr_{i_1},\attr_{i_2}\}} f_2(\attset^\prime)$,
    where $f_2(\attset^\prime)=f_1(\attset^\prime) + f_1(\attset^\prime\cup\{\attr_{i_2}\})$, etc.
\end{proofE}

\subsection*{Tensor View of Query Decomposition}
As some readers find tensor representations of data and queries more intuitive, we  explain
 query decomposition from this viewpoint.
Let $\attset=\{\attr_{i_1},\dots,\attr_{i_\ell}\}$ be a subset of attributes and let $\marginal_\attset$ be a marginal on those attributes, expressed as a tensor. Let $\query$ be a query (in tensor form)  that is answerable from $\marginal_\attset$ (e.g., left part of Fig \ref{fig:decomposition_pre}). Given an $\attset^\prime\subseteq \attset$, we construct $\query_{\deco\attset^\prime}$ from $\query$ as follows:
\begin{enumerate}[leftmargin=*]
    \item Add up the tensor representation of $\query$ over the dimensions not  in $\attset^\prime$, to get an intermediate tensor $\mat{X}_1$. The size of the query representation is $\prod\limits_{\attr_{i_j}\in\attset} \attsize_{i_j}$ while the size of  $\mat{X}_1$ is  $\prod\limits_{\attr_{i_j}\in\attset^\prime} \attsize_{i_j}$.
    \item Multiply $\mat{X}_1$ by the scalar $(\prod_{\attr_{i_j}\in\attset^\prime} \attsize_{i_j})/ (\prod_{\attr_{i_j}\in\attset} \attsize_{i_j})$, to get the intermediate tensor $\mat{X}_2$. Thus $\mat{X}_2$ is equivalent to averaging $\query$ over the dimensions not included in $\attset^\prime$.
    \item $\query_{\deco\attset^\prime}$ is obtained from $\mat{X}_2$ by centering over each dimension.
\end{enumerate}
As an example, consider the query in Figure \ref{fig:decomposition_pre}. To obtain $\query_{\deco\{\}}$, we sum up $\query$ along each dimension to get 3. Since the size (number of cells) of $\query$ is 6 and the size of $\query_{\deco\{\}}$ is 1, we multiply the 3 by 1/6 to get $\query_{\deco\{\}}=1/2$, which is the average of the entries in $\query$.

To obtain $\query_{\deco\{\attr_2\}}$, we sum up along the dimension for $\attr_1$ to get \begin{tabular}{|c|c|c|}\hline
0&1&2\\\hline\end{tabular}. This result has size 3 while $\query$ has size 6, so we multiply by 3/6 to get \begin{tabular}{|c|c|c|}\hline
0&1/2&1\\\hline\end{tabular}. To center this row (so it has mean 0), we need to subtract 1/2 from each entry to get \begin{tabular}{|c|c|c|}\hline
-1/2&0&1/2\\\hline\end{tabular}, which is equal to the $\query_{\deco\{\attr_2\}}$ shown in Figure \ref{fig:decomposition_post}.

Finally, to obtain $\query_{\deco\{\attr_1,\attr_2\}}$ (i.e., $\attset^\prime=\attset$), no summation needs to be performed in Step 1. Step 2 multiplies every entry by 1. Step 3 requires us to iteratively center each row and then each column (one can alternatively center columns first and then rows -- this doesn't change the outcome). We center the first row by subtracting $2/3$ from each entry in row 1, and center the second row by subtracting $1/3$. This gives 
\begin{tabular}{|c|c|c|}\hline
-2/3&1/3&1/3\\\hline
-1/3&-1/3&2/3\\\hline\end{tabular}.
Now each row adds up to 0. To center the columns, we add 1/2 to the first column, 0 to the second, and subtract 1/2 from the third. This yields \begin{tabular}{|c|c|c|}\hline
-1/6&1/3&-1/6\\\hline
1/6&-1/3&1/6\\\hline\end{tabular}, which is $\query_{\deco\{\attr_1,\attr_2\}}$.
Note the rows and columns add up to 0.

Thus, we see that in the tensor view, we average along the dimensions we eliminate, and then center along each dimension we keep so that the average along those dimensions is 0.


\subsection{Optimizing a Workload Partition}\label{sec:solvers}

Once \sysname has performed an orthogonal decomposition of a workload $\workload$ into independent subworkloads $\workload_{\deco\attset^\prime}$, it can optimize each subworkload independently and in parallel. That is, for each
$\workload_{\deco\attset^\prime}$, it needs to create a matrix mechanism that can answer the queries in $\workload_{\deco\attset^\prime}$ at a privacy cost of 1 while minimizing the weighted sum squared error over 
$\workload_{\deco\attset^\prime}$.
\sysname is free to use any technique in the literature to perform this task, even the optimal non-scalable low-dimensional algorithms \cite{yuan2015optimizing} for matrix mechanism construction. The reason is that if all queries $\queryv\in\workload$  only access a few attributes each, then the decomposed queries $\queryv_{\deco\attset^\prime}$ will also access only a few attributes each. The size of the vector representation of $\queryv_{\deco\attset^\prime}$ is $\prod_{\attr_i\in\attset^\prime} |\attr_{i}|$ and this is manageable when $\attset^\prime$ is small (otherwise, one can add approximation algorithms to \sysname as discussed below). 

%
Given a subworkload $\workload_{\deco\attset^\prime}$ that can be answered from a  marginal $\datavec_{\attset^\prime}$, a matrix mechanism constructor \cite{LMHMR15} produces:
%
\begin{itemize}[leftmargin=*]
   \item A gaussian linear mechanism $\mech(\datavec_{\attset^\prime})\equiv \bmat \datavec_{\attset^\prime} + N(\mat{0},\covar)$. This mechanism is run once to get a noisy output $\vec{z}$.
    \item $\recon$: a linear function such that $\recon(\queryv,\vec{z}$) is a noisy answer to $\queryv\in\workload_{\attset^\prime}$. 
    \item $\varfun$: a function such that $\varfun(\queryv)$  is the variance of the reconstructed answer (i.e. Variance$(\recon(\queryv,\vec{z}))$. This variance is data-independent hence $\varfun$ does not need any data  access.
\end{itemize}

\subsubsection*{Optimal Matrix Mechanism Construction}
Let $\mat{W}$ be the matrix whose rows are the query vectors in $\workload_{\deco\attset^\prime}$.
Let $\mat{D}$ be a diagonal matrix, where entry $\mat{D}[i,i]$ is the weight of the query that is in the $i^\text{th}$ row of $\mat{W}$. When the goal is to minimize the weighted sum squared error of the queries,
Yuan et al. \cite{yuan2015optimizing} formulated matrix mechanism construction in terms of the following semidefinite optimization problem (where $\dagger$ denotes pseudoinverse):
\begin{align}
    \arg\min_{\mat{V}}\quad&\text{trace}(\mat{W}^\top\mat{D}\mat{W}\mat{V}^{\dagger})\label{eqn:semidef}\\
    &\text{s.t. }\mat{V} \text{ is symmetric positive semidefinite }\nonumber\\
    &\text{the row space of }\mat{V} \text{ contains the row space of }\mat{W}\nonumber\\
    &\text{and }\mat{V}[i,i]\leq 1 \text{ for all }i\nonumber
\end{align}
which can be solved using the algorithm of Yuan et al. \cite{yuan2015optimizing} or a simplified version
by McKenna et al. \cite{mckenna2018optimizing} or even off-the-shelf semidefinite program solvers. Once a solution $\mat{V}$ is obtained, the rest of the steps are:
\begin{enumerate}[leftmargin=*]
\item Compute the Cholesky decomposition $\mat{V}=\bmat^T\bmat$ and returns $\bmat$. Optionally, project the rows of $\bmat$ onto the rowspace of $\mat{W}$. 
\item Given a desired privacy cost $\pcostbound$, set $\sigma^2=\max_i (\bmat^\top\bmat)[i,i]/\pcostbound$. Set  $\mech(\datavec_{\attset^\prime}) \equiv \bmat \datavec_{\attset^\prime} + N(0, \sigma^2\identity)$.
\item Define $\recon(\queryv,\vec{z})\equiv\queryv\bmat^\dagger\vec{z}$, where $\bmat^\dagger$ is the pseudoinverse.
\item Define $\varfun(\queryv)\equiv \sigma^2\queryv\bmat^\dagger(\bmat^\dagger)^\top\queryv^\top$ (recall $\queryv$ is a row vector).
\end{enumerate}

\subsubsection*{Approximate Matrix Mechanism Construction}
Sometimes the semidefinite program in Equation \ref{eqn:semidef} is difficult to solve because the main
quantity $\mat{W}^\top\mat{D}\mat{W}$ is very large. This can happen when $\attset^\prime$
has many variables (i.e., some queries do access many variables) or when some attributes $\attr_i\in\attset^\prime$
already have large domains. The flexibility of \sysname means that in those situations, approximate algorithms
for creating matrix mechanisms can be added to \sysname to handle those cases. Common approaches include (1) not running the solver for Equation \ref{eqn:semidef}
to completion, or (2) imposing additional structure on the optimization variables $\mat{V}$. For example, one common
approach used in practice is to fix a basis $\mat{U}$ whose rowspan includes the rows of $\mat{W}$, and require that $\mat{V}=\mat{U}^T\mat{H}\mat{U}$
for some diagonal matrix $\mat{H}$ \cite{privelet,dawa,LMHMR15}. This results in fewer optimization parameters (the diagonal entries instead of a full matrix)
and can often be solved in close form.
This flexibility and modularity is responsible for the generality of \sysname. In fact, we show in Section \ref{sec:accuracy} that using approximate (instead of optimal) techniques can make \sysname equivalent to
RP, RP+, and WFF.

\subsection{Assembling the Overall Solution.}\label{sec:rescale}

After \sysname is done constructing a matrix mechanism for each nonempty subworkload $\workload_{\deco\attset^\prime}$, it ends up with a linear gaussian mechanism for each $\attset^\prime$,
having the form $\mech_{\attset^\prime}(\datavec_{\attset^\prime})\equiv \bmat_{\attset^\prime} \datavec_{\attset^\prime} + N(\mat{0}, \covar_{\attset^\prime})$. Each such mechanism has privacy cost 1. Let $L_{\attset^\prime}$ be the weighted sum squared error of the matrix mechanism over the subworkload  $\workload_{\deco\attset^\prime}$.

If \sysname rescales the noise used by $\mech_{\attset^\prime}$ by some value $\sigma^2_{\attset^\prime}$
(i.e., redefining it as
$\mech_{\attset^\prime}(\datavec_{\attset^\prime})\equiv \bmat_{\attset^\prime} \datavec_{\attset^\prime} + N(\mat{0}, \sigma^2_{\attset^\prime}\covar_{\attset^\prime})$)
then the privacy cost becomes $1/\sigma^2_{\attset^\prime}$ and the weighted sum-squared error over the subworkload becomes $\sigma^2_{\attset^\prime} L_{\attset^\prime}$.

Thus, \sysname searches for values of $\sigma^2_{\attset^\prime}$, for each $\attset^\prime$, so that the overall privacy cost is $\pcostbound$, while making the total weighted sum squared error as small as possible. That is, it needs to solve the following optimization problem:
\begin{align*}
    \min_{\sigma^2_{\attset^\prime}} \sum_{\attset^\prime} \sigma^2_{\attset^\prime} L_{\attset^\prime}\quad
                                     \text{s.t., }\sum_{\attset^\prime} \frac{1}{\sigma^2_{\attset^\prime}}=\pcostbound
\end{align*}
The form of this optimization problem was studied in \cite{rplus} (Lemma 3) and the optimal choices of the $\sigma^2_{\attset^\prime}$ are:
\begin{align*}
    \sigma^2_{\attset^\prime} &= \frac{1}{\pcostbound} \frac{\sum_{\attset^{\prime\prime}}\sqrt{L_{\attset^{\prime\prime}}}}{\sqrt{L_{\attset^\prime}}}
\end{align*}
These scale factors are computed by Algorithm \ref{alg:smasher} in Line \ref{line:rescale2}. This discussion proves the privacy properties of \sysname:
\begin{theorem}
    The privacy cost of Algorithm \ref{alg:smasher} is $\leq$  the input parameter $\pcostbound$.
\end{theorem}

After these calculations have been performed, the measure phase (which is the only step that accesses raw data, and hence the only step that incurs the privacy cost) can begin.
\sysname runs all of the linear gaussian mechanisms $\mech_{\attset^\prime}(\datavec_{\attset^\prime})= \bmat_{\attset^\prime} \datavec_{\attset^\prime} + N(\mat{0}, \sigma^2_{\attset^\prime}\covar_{\attset^\prime})$
with this appropriately scaled noise. It collects their respective noisy answers $\vec{z}_{\attset^\prime}$
and does not access the data anymore.

To answer a query $\queryv\in \workload$, \sysname re-uses the noisy answers it has collected.
If $\queryv$ accesses attributes in the set $\attset$, then \sysname decomposes it into subqueries $\queryv_{\deco\attset^\prime}$ for all $\attset^\prime\subseteq\attset$ using Equation \ref{eq:vecdecomp} from Section \ref{sec:vectordeco}. It reconstructs an answer to each $\queryv_{\deco\attset^\prime}$ using the reconstruction function $\recon_{\attset^\prime}$ provided by the matrix mechanism for the subworkload $\workload_{\deco\attset^\prime}$ (i.e., \sysname computes $\recon_{\attset^\prime}(\queryv_{\deco\attset^\prime}, \vec{z}_{\attset^\prime})$). Since the answer to $\queryv$ is the sum of the answers of its subqueries (Theorem \ref{thm:decomp}), the answer to $\queryv$ is calculated as:
\begin{align*}
    \text{answer}(\queryv) &= \sum_{\attset^\prime\subseteq\attset} \recon_{\attset^\prime}(\queryv_{\deco\attset^\prime}, \vec{z}_{\attset^\prime})
\end{align*}

\section{Accuracy of {\protect\sysname}}\label{sec:accuracy}
Next we show that algorithms equivalent to RP, RP+, and WFF can be obtained as special cases when \sysname uses approximate rather than optimal matrix mechanism construction
techniques, which implies that \sysname would dominate those methods on every workload, if it uses optimal solvers for the subworkloads. We do this by abstracting some common linear algebraic properties that they share. The first property is
respecting residual spaces, which we define as follows:

\begin{definition}\label{def:respectresid}
A gaussian linear mechanism $\mech$ \textbf{respects residual spaces} if there exists an $\attset=\{\attr_{i_1},\dots, \attr_{i_\ell}\}$ such that:
\begin{itemize}[leftmargin=*]
    \item $\mech(\datavec_{\attset})\equiv \bmat \datavec_{\attset} + N(\mat{0}, \covar)$ for some $\bmat$ and $\covar$ (i.e., $\mech$ adds Gaussian noise to linear queries computable from the marginal on $\attset$).
    \item Each row of $\bmat$ belongs to the row space of $$\left(\identity_{\attsize_{i_1}}-\frac{1}{\attsize_{i_1}}\one_{\attsize_{i_1}}\one^\top_{\attsize_{i_1}}\right)\kron \cdots \kron \left(\identity_{\attsize_{i_\ell}}-\frac{1}{\attsize_{i_\ell}}\one_{\attsize_{i_\ell}}\one^\top_{\attsize_{i_\ell}}\right).$$
\end{itemize}
\end{definition}
Intuitively, this means that rows of $\bmat$ belong to the same linear space as queries in the chosen subworkload $\workload_{\deco\attset}$ and hence $\mech$ could be used to provide answers to them.

\begin{example}
    ResidualPlanner (RP) \cite{rplanner} uses gaussian linear mechanisms of the form: $\mech(\datavec_{\attset})\equiv \bmat \datavec_{\attset} + N(\mat{0}, \covar)$, where 
    the matrix $\bmat=\textbf{Sub}_{\attsize_{i_1}}\kron\cdots\kron \textbf{Sub}_{\attsize_{i_\ell}}$, where $\textbf{Sub}_{\attsize}$ is a $(\attsize-1)\times\attsize$ matrix that can be written in block matrix form as $[\one_{\attsize}\quad-\identity_{\attsize}]$ (the first column is full of ones, entries of the form $(i, i+1)$ contain -1, and everything else is 0). For example $\textbf{Sub}_3=\left(\begin{smallmatrix}1 & -1 & 0\\ 1 & 0 & -1\end{smallmatrix}\right)$. It is easy to verify that the row space of each $\textbf{Sub}_{\attsize_{i_j}}$ is a subspace of the rowspace of $\left(\identity_{\attsize_{i_j}}-\frac{1}{\attsize_{i_j}}\one_{\attsize_{i_j}}\one^\top_{\attsize_{i_j}}\right)$ and so $\bmat$ satisfies the condition of Definition \ref{def:respectresid}. Therefore every gaussian linear mechanism used by RP respects residual spaces.
\end{example}

\begin{example}
    ResidualPlanner+ (RP+) \cite{rplus} uses gaussian linear mechanisms of the form: $\mech(\datavec_{\attset})\equiv \bmat \datavec_{\attset} + N(\mat{0}, \covar)$, where 
    the matrix $\bmat=\mat{S}_{1}\kron\cdots\kron\mat{S}_\ell$, where each $\mat{S}_i$ depends on a user's input and is modified so that its row space belongs to the row space of $\left(\identity_{\attsize_{i_j}}-\frac{1}{\attsize_{i_j}}\one_{\attsize_{i_j}}\one^\top_{\attsize_{i_j}}\right)$. Hence RP+ also respects residual spaces. 
\end{example}

\begin{example}
    Weighted Fourier Factorization (WFF) \cite{lebeda2025weightedfourierfactorizationsoptimal} applies the multidimensional Fourier transform to different data marginals and adds independent noise to their coefficients. This is equivalent to creating linear mechanisms of the form $\mech(\datavec_{\attset})\equiv \bmat \datavec_{\attset} + N(\mat{0}, \covar)$ where $\covar$ is a diagonal matrix and each row of $\bmat$ is created by taking the real and imaginary components of vectors of the form $\vec{v}_1 \kron \vec{v}_2 \kron \cdots \kron \vec{v}_\ell$, where each $\vec{v}_j$ has size $\attsize_{i_j}$ and has the form:
    \begin{align*}
        \vec{v}_j&=\Bigl[\exp\!\bigl( \tfrac{-2\pi\imag t}{\attsize_{i_j}}\cdot 0\bigr),\;
        \exp\!\bigl( \tfrac{-2\pi\imag t}{\attsize_{i_j}}\cdot 1\bigr),\\
        &\qquad\quad \dots,\;
        \exp\!\bigl( \tfrac{-2\pi\imag t}{\attsize_{i_j}}\cdot(\attsize_{i_j}-1)\bigr)\Bigr]
    \end{align*}
    where $\imag=\sqrt{-1}$ and $t$ is an integer between $1$ and $|\attsize_{i_j}|-1$. Varying the values of $t$ inside  each of the $\vec{v}_j$ gives the different rows of $\bmat$.
    Although this is a bit complex, one can see that each $\vec{v}_j$ is orthogonal to $\one_{\attsize_{i_j}}$ and therefore $\bmat$ satisfies the conditions of Definition \ref{def:respectresid}.
\end{example}

The next condition is that if a matrix mechanism is given an overall target privacy cost of $\pcostbound$, then the privacy costs of the gaussian linear mechanisms $\mech_1,\dots, \mech_m$ it uses add up to $\pcostbound$. It is not a vacuous condition because if $\mech_i$ uses strategy matrix $\bmat_i$ and covariance matrix $\covar_i$ then its privacy cost is $\max_j \bmat_i^\top\covar_i^{-1}\bmat_i)[j,j]$ (Definition \ref{def:glm}). However, the overall privacy cost (Theorem \ref{thm:overall}) is:
\begin{align}
    \pcostbound &= \max_j \left(\sum_{i=1}^m \bmat_i^\top\covar_i^{-1}\bmat_i\right)\![j,j] \nonumber\\
    &\leq \sum_{i=1}^m \max_j \left(\bmat_i^\top\covar_i^{-1}\bmat_i\right)\![j,j]\label{eq:splittable}\\
    &=\sum_{i=1}^m \pcost(\mech_i)\nonumber
\end{align}
Thus the overall privacy cost, in general, can be smaller than the sum of the individual privacy costs of the mechanisms.

\begin{definition}\label{def:splittable}
    A matrix mechanism is \textbf{splittable} if, when given a target privacy cost bound $\pcostbound$, it creates a collection of gaussian linear mechanisms and the sum of their individual privacy costs is $\pcostbound$.
\end{definition}

The RP+ algorithm is explicitly designed to be splittable. The linear mechanisms used by RP all have constant diagonals for $\bmat_i^\top\covar^{-1}\bmat_i$ and so clearly Equation \ref{eq:splittable} holds with equality, so RP is splittable. When WFF adds noise to Fourier coefficients, the complex part of the coefficient gets the same variance as the real part of the coefficient. As a result, the privacy cost matrix of each gaussian linear mechanism it uses has a constant diagonal and therefore, like RP, it is splittable.

Our final theoretical result is that if a competing matrix mechanism  is splittable and every gaussian linear mechanism it uses respects residual spaces then the weighted sum squared error over $\workload$ that it achieves cannot be smaller than the weighted sum squared error of \sysname (provided that \sysname uses an optimal solver for Equation \ref{eqn:semidef}). In particular each gaussian linear mechanisms it uses is a feasible (and potentially suboptimal) solution to the matrix mechanism construction semidefinite program (Equation \ref{eqn:semidef}) for some subworkload $\workload_{\deco\attset}$.

\begin{theoremE}[][category=accuracy]\label{thm:accuracy}
Given a workload $\workload$ of linear queries that can be answered from data marginals, and a privacy cost bound $\pcostbound$, let \textbf{Compete} denote a matrix mechanism that can answer the queries in $\workload$, is splittable, and which uses gaussian linear mechanisms that all respect residual spaces. Then:
\begin{itemize}[leftmargin=*]
    \item The weighted sum-squared error of \textbf{Compete} over $\workload$ is equal to the sum of the weighted sum-squared errors of $\textbf{Compete}$ on the subworkloads created by \sysname.
    \item For any given $\attset$, let $\mech_1,\dots, \mech_m$ be the subset of Gaussian linear mechanisms of \textbf{Compete} that are answerable from the marginal $\datavec_\attset$. Then $\mech_1,\dots, \mech_m$ together correspond to a feasible (but potentially suboptimal) solution to the optimization problem in Equation \ref{eqn:semidef} for subworkload $\workload_{\deco\attset}$ (i.e., the $\mat{W}$ in that equation is formed from the query vectors in $\workload_{\deco\attset}$).
    \item When \sysname uses solvers that solve Equation \ref{eqn:semidef} to optimality, the weighted sum squared error of \sysname over $\workload$ is $\leq$ the weighted sum-squared error of \textbf{Compete} over $\workload$.
\end{itemize}
\end{theoremE}
\begin{proof}[Proof sketch] (see \shortlong{the full version \cite{qsarxiv}}{the appendix} for the full proof).
The main steps are to show that for any non-empty subworkload generated by \sysname, a subset of the strategy queries used by \textbf{Compete} form a linear basis for the space spanned by that subworkload. Furthermore, the strategy queries from \textbf{Compete} that are assignable to one subworkload are linearly independent from those assignable to a different subworkload. Then, like with Theorem \ref{thm:decomp}, answering a query $\queryv$ using \textbf{Compete} is equivalent to using \textbf{Compete} to answer the subqueries that $\queryv$ is decomposed into, and adding up the answers. Since \sysname is optimal for the subqueries, it then has lower error than \textbf{Compete}.
\end{proof}
\begin{proofE}
    First, we note that if $\mech_1,\dots,\mech_m$ are answerable from marginal $\datavec_\attset$ then they can be turned into a single matrix mechanism $\mech_\attset$ that is answerable from $\datavec_\attset$. Indeed, let $\bmat_1,\dots,\bmat_m$ be their respective strategy matrices and $\covar_1,\dots,\covar_m$ be their covariance matrices. Form the matrix $\bmat_\attset$ by stacking the $\bmat_i$:
    \begin{align*}
        \bmat_\attset &=\left(
        \begin{matrix}
            \bmat_1\\ \bmat_2\\ \vdots \\ \bmat_m
        \end{matrix}
        \right)
    \end{align*}
    and form $\covar_\attset$ as the block diagonal matrix:
    \begin{align*}
        \covar_\attset &=\left(
        \begin{matrix}
            \covar_1 & \mat{0} &\cdots & \mat{0}\\
            \mat{0} & \covar_2 &\ddots & \mat{0}\\
            \vdots & \vdots & \ddots & \vdots\\
            \mat{0} & \mat{0} & \cdots &\covar_m
        \end{matrix}
        \right)
    \end{align*}
    and define $\mech_\attset(\datavec_\attset)\equiv \bmat_\attset \datavec_\attset + N(\mat{0},\covar_\attset)$. Then it is easy to see that the distribution of the output
    of $\mech_\attset$ is the same as the joint distribution of the outputs of $\mech_1,\dots, \mech_m$.

    Our next observation is that we can convert $\mech_\attset$ into a mechanism that has a diagonal covariance matrix. Define $\bmat^*_\attset=\covar_\attset^{-1/2}\bmat_\attset $, where $\covar_\attset^{1/2}$ is the matrix square root ($\covar_\attset^{1/2}\covar_\attset^{1/2}=\covar_\attset$). $\covar_\attset^{1/2}$ exists because $\covar$ is symmetric positive definite as it is a covariance matrix for Gaussian noise. 

    Define the gaussian linear mechanism $\mech^*_\attset(\datavec_\attset)\equiv\bmat^*_\attset\datavec_\attset + N(\mat{0},\identity)$. Then it is easy to check that:
    \begin{itemize}[leftmargin=*]
        \item The output distributions of $\mech_\attset(\datavec_\attset)$ and $\covar^{1/2}_\attset \mech^*_\attset(\datavec_\attset)$ (i.e., run $\mech^*_\attset$ and multiply the result by $\covar^{1/2}$) are the same.
        \item The output distributions of $\covar^{-1/2}_\attset\mech_\attset(\datavec)$ and $\mech^*_\attset(\datavec)$ are the same.
    \end{itemize}
    Thus $\mech_\attset$ and $\mech^*_\attset$ can be obtained from each other by postprocessing, and so by the postprocessing invariance of differential privacy, both mechanisms have the same privacy cost. Furthermore, in whatever situation one needs to use the output of $\mech$ (e.g., in the reconstruction function), one can instead use the output of $\mech^*$ and left multiply it by $\covar^{-1/2}_\attset$.

    Next, every $\mech^*_\attset$, whose input should be $\datavec_\attset$ can be converted into a mechanism $\mech^{full}_\attset$ that operates on the full data $\datavec$. Specifically, define:
    \begin{align*}
    \mquery{\attset} &= \mat{V}_1 \kron \mat{V}_2\kron \cdots \kron \mat{V}_\numattr\\
    \text{where } \mat{V}_i &=
    \begin{cases}
        \identity_{\attsize_i} & \text{ if }\attr_i\in\attset\\
        \one^T_{\attsize_i} &\text{ otherwise }
    \end{cases}\\
    \mech^{full}_\attset(\datavec) &= \bmat^*_\attset \mquery{\attset}\datavec + N(\mat{0}, \identity)
    \end{align*}
    and similarly any query $\queryv$ over $\datavec$ can be converted into a query $\queryv^{full}$ over $\data$ as follows:
    $\queryv^{full}=\queryv \mquery{\attset}$.
    (note $\queryv\cdot\datavec_\attset=\queryv^{full}\cdot\datavec$).

\vspace{1em}\noindent\textbf{Step 1: optimal reconstruction.}
We next consider what is the optimal linear reconstruction function.
First, note that  the row spaces of $\bmat_\attset$ and $\bmat^*_\attset$ are the same, and since residual subsapces are respected, those row spaces are
 are contained in the row space of 
\begin{multline*}
Proj_\attset\equiv\left(\identity_{\attsize_{i_1}}-\tfrac{1}{\attsize_{i_1}}\one_{\attsize_{i_1}}\one^\top_{\attsize_{i_1}}\right)\kron \cdots \\
\kron \left(\identity_{\attsize_{i_\ell}}-\tfrac{1}{\attsize_{i_\ell}}\one_{\attsize_{i_\ell}}\one^\top_{\attsize_{i_\ell}}\right).
\end{multline*}
where $\attset=\{\attr_{i_1},\dots, \attr_{i_\ell}\}$.

Since $Proj_\attset$ is an idempotent matrix ($Proj_\attset Proj_\attset=Proj_\attset$), it means that $\bmat^*_\attset Proj_\attset=\bmat^*_\attset$.

Define
\begin{align*}
    \mat{Q}^{resid}_\attset &\equiv Proj_\attset \mquery{\attset}
    = \mat{V}_1 \kron \cdots \kron \mat{V}_\numattr\\
    \text{where } \mat{V}_i &=
    \begin{cases}
        \identity_{\attsize_i}-\tfrac{1}{\attsize_i}\one_{\attsize_i}\one^\top_{\attsize_i} & \text{if }\attr_i\in\attset\\
        \one^T_{\attsize_i} &\text{otherwise}
    \end{cases}
\end{align*}
and this means $\mat{Q}^{resid}_\attset (\mat{Q}^{resid}_{\attset^\prime})^\top=\mat{0}$ when $\attset^\prime\neq\attset$.

Combining these facts, for any $\attset^\prime\neq\attset$,
\begin{align*}
    &(\bmat^*_\attset \mquery{\attset})(\bmat^*_{\attset^\prime} \mquery{\attset^\prime})^\top \\
    &=(\bmat^*_\attset Proj_\attset\mquery{\attset})(\bmat^*_{\attset^\prime}Proj_{\attset^\prime} \mquery{\attset^\prime})^\top\\
    &=(\bmat^*_\attset \mat{Q}^{resid}_{\attset})(\bmat^*_{\attset^\prime} \mat{Q}^{resid}_{\attset^\prime})^\top\\
    &=\mat{0}
\end{align*}
Thus $\mech^{full}_\attset$ and $\mech^{full}_{\attset^\prime}$ use orthogonal (hence independent) strategy matrices and hence access orthogonal information from $\datavec$.

Next, if \sysname generates a non-empty subworkload  $\workload_{\deco\attset}$ then \textbf{Compete} must have a corresponding $\mech^{full}_\attset$ (since a non-empty subworkload $\workload_{\deco\attset}$ implies $\queryv^{full}$ for some $\queryv\in\workload$ must have a component in a direction orthogonal to all of the strategy matrices of $\mech^{full}_{\attset^\prime}$ for $\attset^\prime\neq\attset$, since the mechanisms of \textbf{Compete} respect residual spaces). Furthermore, if for some $\attset$ there exists an \textbf{empty} subworkload  $\workload_{\deco\attset}$, then there shouldn't be any $\mech^{full}_\attset$ as its strategy matrix would be orthogonal to all queries (hence running such a mechanism would consume privacy budget for no improvement in accuracy).

Hence, define \textbf{Good} to be the set of all $\attset$ for which $\workload_{\deco\attset}$ is nonempty. The previous discussion means that \textbf{Compete} has a mechanism $\mech^{full}_{\attset}$ if and only if $\attset\in\textbf{Good}$.

Since the $\mech^{full}_{\attset}$ have independent noise and independent strategy matrices, then given their noisy outputs $\vec{z}_\attset = \mech^{full}_{\attset}(\datavec)$ for $\attset\in\textbf{Good}$, the optimal linear reconstruction of the answer to a $\queryv\in \workload$ that is answerable from some $\datavec_{\attset^*}$ is the query times the sum of the strategy pseudoinverse times the noisy answers:
\begin{align*}
    &\queryv^{full}\left(\sum_{\attset\in\textbf{Good}} (\bmat^*_\attset \mquery{\attset})^\dagger\vec{z}_\attset\right)\\
    &\queryv\mquery{\attset^*}\left(\sum_{\attset\in\textbf{Good}} (\bmat^*_\attset \mquery{\attset})^\dagger\vec{z}_\attset\right)\\
     &\queryv\mquery{\attset^*}\left(\sum_{\attset\in\textbf{Good}} (\bmat^*_\attset \mat{Q}^{resid}_{\attset})^\dagger\vec{z}_\attset\right)\\
     &=\sum_{\attset\in\textbf{Good}}\queryv\mquery{\attset^*}(\mat{Q}^{resid}_{\attset})^\dagger(\bmat^*_\attset)^\dagger\vec{z}_\attset
\end{align*}
Next, we note that
\begin{align*}
    (\mat{Q}^{resid}_\attset)^\dagger
    &= \mat{V}_1 \kron \cdots \kron \mat{V}_\numattr\\
    \text{where } \mat{V}_i &=
    \begin{cases}
        \identity_{\attsize_i}-\tfrac{1}{\attsize_i}\one_{\attsize_i}\one^\top_{\attsize_i} & \text{if }\attr_i\in\attset\\
        \tfrac{1}{\attsize_i}\one_{\attsize_i} &\text{otherwise}
    \end{cases}
\end{align*}
so
\begin{align*}
    \mquery{\attset^*}&(\mat{Q}^{resid}_{\attset})^\dagger =\mat{W}_1\kron\cdots\kron\mat{W}_\numattr\\
    \text{where }& \mat{W}_i=
    \begin{cases}
        1 & \text{if }\attr_i\notin \attset^*\!\cup\attset\\
        \mat{0} & \text{if }\attr_i\in \attset\setminus \attset^*\\
        \tfrac{1}{\attsize_i}\one_{\attsize_i} & \text{if }\attr_i\in\attset^*\setminus\attset\\
        \identity_{\attsize_i}-\tfrac{1}{\attsize_i}\one_{\attsize_i}\one^\top_{\attsize_i} & \text{if }\attr_i\in\attset^*\cap\attset
    \end{cases}
\end{align*}
and so $\mquery{\attset^*}(\mat{Q}^{resid}_{\attset})^\dagger$ equals the 0 matrix if $\attset\not\subseteq\attset^*$
and otherwise is the matrix from Equation \ref{eq:vecdecomp} used to project $\queryv$ into $\queryv_{\deco\attset}$.

Hence, the reconstruction of the answer to a $\queryv\in \workload$ that is answerable from some $\datavec_{\attset^*}$
is 
\begin{align*}
    &\queryv^{full}\left(\sum_{\attset\in\textbf{Good}} (\bmat^*_\attset \mquery{\attset})^\dagger\vec{z}_\attset\right)\\
     &=\sum_{\attset\in\textbf{Good}}\queryv\mquery{\attset^*}(\mat{Q}^{resid}_{\attset})^\dagger(\bmat^*_\attset)^\dagger\vec{z}_\attset\\
     &=\sum_{\attset\subseteq\attset^*}\queryv_{\deco\attset}(\bmat^*_\attset)^\dagger\vec{z}_\attset
\end{align*}
we also note that $\queryv_{\deco\attset}(\bmat^*_\attset)^\dagger\vec{z}_\attset$ is the optimal reconstruction
(based on the pseudoinverse) of $\queryv_{\deco\attset}$ based on the output $\vec{z}_\attset$ of $\mech^*_{\attset}$.

Thus for \textbf{Compete}, just like for \sysname, the reconstructed answer to $\queryv$ is the sum of the reconstructed answers to its projected subqueries. Furthermore, since each projected subquery is answered using noise that is independent of the other subqueries (because it comes from a separate gaussian linear mechanism), the variance of $\queryv$ is the sum of the variances of the subqueries. Since $\queryv$ and its subqueries also have the same weights, then the weighted sum squared error of reconstructed answers to queries in $\workload$ is equal to the sum of the weighted sum-squared errors of reconstructed answers to the subworkloads. 

This proves the first claim of the theorem.

\vspace{1em}\noindent\textbf{Step 2: separate subworkload optimization.}
The previous results showed that the role of $\mech^*_\attset$ is to serve as a gaussian linear mechanism
from whose output the answers to queries in $\workload_{\deco\attset}$ are reconstructed. If $\beta_\attset$ is the privacy
cost of $\mech^*_\attset$ then rescaling the entries of $\bmat^*_\attset$  by $1/\sqrt{\beta_\attset}$ results in a gaussian linear mechanism whose privacy cost is 1. Hence $(\bmat^*_\attset)^T(\bmat^*_\attset)/\beta_\attset$ is a feasible solution to the matrix $\mat{V}$ in Equation \ref{eqn:semidef}. This proves the second claim.

Hence, if we obtain the optimal matrix mechanism for the subworkload $\workload_{\deco\attset}$ and rescale its noise variance by $1/\beta_{\attset}$ (so that it has privacy cost $\beta_\attset$), and then use that in place of $\mech^*_\attset$, the weighted sum-squared error over the subworkload cannot increase. If we then replace the $1/\beta_{\attset}$ rescaling by the optimal rescaling used by \sysname, the overall weighted sum-squared error cannot increase. This proves the third claim.

\end{proofE}


\section{Experiments}\label{sec:experiments}
%

In this section, we present our experimental results.

\textbf{Baselines.} We compare
\sysname (\textbf{QS}) to ResidualPlanner (\textbf{RP}) \cite{rplanner}, ResidualPlanner+ (\textbf{RP+}) \cite{rplus}, \textbf{HDMM} \cite{mckenna2018optimizing}, and
Weighted Fourier Factorization (\textbf{WFF}) \cite{lebeda2025weightedfourierfactorizationsoptimal}. Public implementations of prior work do not support all of the query workloads we consider, so we take the following approach. RP only supports marginals but is completely subsumed by RP+ (hence comparison to RP can be dropped). Furthermore, we modifed the RP+ implementation to support a wider set of queries. WFF does not have a public implementation so we re-implemented
it (a practical application of Theorem \ref{thm:accuracy}). We did not modify HDMM code and so query workloads that are not supported by the implementation are marked as NI (not implemented).

\textbf{Implementation details.}
We implemented \sysname in Python~3.11 using
NumPy, SciPy, and PyTorch.
All computations are CPU-only; our experiments are run on a dual-socket
AMD EPYC 7763 server (128 cores, 1\,TB RAM) running Linux and the
residual spaces are optimized sequentially.
%

\subsection{Workload queries.}
We evaluate the matrix mechanisms on a wider variety of queries than prior work. Specifically, we consider the following:
%
\begin{itemize}[nosep,leftmargin=*]
  \item \textbf{Marginal} queries. One-way marginals have the form ``how many records have $\attr_i=c$'' (for each choice of c), two-way marginals have the form ``how many records have $\attr_i=c_1$ and $\attr_j=c_2$'', etc. Note that RP, RP+, WFF, and \sysname are all optimal for pure marginal workloads, so this serves as a sanity check.
  \item \textbf{Prefix}-sum queries: these are queries of the form ``how many records have $\attr_i\leq c$'' for different values of $c$. Two-dimensional prefix-sum queries have the form ``how many records have $\attr_{i}\leq c_1$ and $\attr_{j}\leq c_2$,'' etc.
  \item \textbf{Range} queries: these are axis-aligned queries of the form ``how many records have $\attr_i\in [c_1, c_2]$'' and are extended to multiple dimensions in the obvious way  (e.g., ``how many records have $\attr_i\in [c_1, c_2]$ and $\attr_j\in [c_3, c_4]$'').
  \item\textbf{Circular} range queries: these are range queries that wrap around like on a clock (e.g., ``how many checkins are between ``10pm and 1am'') and are extended to multiple dimensions in the obvious way. WFF \cite{lebeda2025weightedfourierfactorizationsoptimal} is provably optimal for these queries and therefore so is \sysname.
  \item \textbf{Affine} queries: these are two-dimensional queries that are \textbf{not} axis-aligned. They are answerable from two-way marginals and have the form ``how many records have $\attr_i+\attr_j\leq c$'' for different choices of $c$. These queries allow for cross-attribute comparisons, which prior work did not consider.
  \item \textbf{Abs} queries: these are also two-dimensional queries that are \textbf{not} axis-aligned. They allow cross-attribute comparisons and have the form ``how many records have $|\attr_i-\attr_j|\leq c$'' for different choices of $c$.
  \item \textbf{Random}: finally we consider random counting queries that are formed by taking a query vector $\queryv$, initially filled with zeros, and then independently flipping each bit with probability $p$. This query type simulates queries without any particular structure. Noting that there is no universally accepted choice of $p$, we chose  $p=0.3$ to get somewhat sparse query vectors.
\end{itemize}
%

\subsection{Metrics}
In our experiments, each matrix mechanism is required to optimize its workload at a privacy cost $\pcostbound=1$. We  measure the root-mean-squared error:

\begin{align*}
    RMSE &=\sqrt{\frac{\sum\limits_{\queryv\in\workload}\varfun(\queryv)}{|\workload|}}
\end{align*}

\subsection{Executive Summary}\label{sec:executive}
For the accuracy experiments, we evaluate all methods on real data in Section \ref{sec:exp:real}
and provide more extensive experiments on synthetic data in Section \ref{sec:exp:synthetic}. Scaling
experiments are provided in Section \ref{sec:exp:scale}.
The high-level summary of  our observations from those sections are:
\begin{itemize}[leftmargin=*]
    \item Among the baselines (RP+, WFF, HDMM), each has its own specialty: HDMM performs very well on marginal queries, RP+ performs well on marginals and prefix sums. WFF is provably optimal for marginals and circular range queries \cite{lebeda2025weightedfourierfactorizationsoptimal}; it also performs well for range queries and random queries (but struggles with prefix queries).
    \item \sysname consistently has the lowest RMSE for all workloads. It matches WFF in places where it is optimal (marginals and circular range queries) and is strictly better than all methods on the rest of the queries.
    \item None of the competitors (RP+, WFF, HDMM) perform well on \textbf{affine} and \textbf{abs} queries, and \sysname shows significant improvements for those query types. It also outperforms HDMM and WFF on prefix queries by a large margin and outperforms RP+ on prefix queries by a smaller margin.
    \item In terms of scaling, the dominant cost arises in optimizing a subworkload, as deploying the optimal solver requires solving a semi-definite optimization program. The resource requirements only grow with the size of the output (e.g., size of the workload query answers) rather than with the domain size (which is exponentially large). To put it into context, consider the following two results from Section \ref{sec:exp:scale}.
    \begin{itemize}[leftmargin=*]
    \item Consider a complex workload consisting of all 1-d range queries, all 2-d affine queries, and all 3-d prefix queries.
    Consider a synthetic dataset with 10 attributes, with each attribute having a domain size of 40. Thus, the overall domain size of a tuple is $10^{40}$. For the workload on such a dataset, there are ${10 \choose 3}=120$ groups of 3-d prefix queries (i.e., $120$ ways of choosing 3 attributes), ${10 \choose 2}=45$ groups of 2-d affine queries, and ${10\choose 1}=10$ groups of 1-d range queries. Hence, answering this workload, even in the non-private setting, requires computing $120+45+10=175$ views of the dataset. For \sysname the solver time was 60 seconds and required 3.4\,GB of memory.
    \item If the number of dimensions is increased to 40, the overall domain size of a tuple becomes $40^{40}\approx 1.2\times 10^{64}$. Now there are  ${40 \choose 3}=9880$ groups of 3-d prefix queries (i.e., $9880$ ways of choosing 3 attributes), ${40 \choose 2}=780$ groups of 2-d affine queries, and ${40\choose 1}=40$ groups of 1-d range queries. Hence answering this workload requires computing $9880 + 780 + 40 = 10700$ views of the data. This is now a $\frac{10700}{175}\approx 61$ times larger problem. The solve time is 3032 seconds (which is only $\frac{3032}{60}\approx 51$ times larger) and the memory requirement is 227\,GB (which is only $\frac{227}{3.4}\approx 67$ times larger).
    \end{itemize}
    We note that in addition to parallelizing the solve phase of \sysname, one can also replace the optimal low-dimensional subworkload optimizer (Section \ref{sec:solvers}) with an approximate optimizer. Hence, any future research improvements in low-dimensional matrix mechanisms will directly benefit \sysname as well. We next consider the experiments in detail.
\end{itemize}



\subsection{Utility Experiments for Real Data}\label{sec:exp:real}

Following prior work~\cite{rplus,mckenna2018optimizing}, we consider the following real dataset schemas.
\textbf{Adult} (9 categorical attributes with domain sizes 42, 16, 15, 9, 7, 6, 5, 2, 2 and 5 numeric attributes with domain sizes 100, 100, 100, 99, 85) \cite{adult_2}, \textbf{CPS} (3 categorical attributes with domain sizes 7, 4, 2 and 2 numeric attributes with sizes 50, 100) \cite{mckenna2018optimizing}, and \textbf{Loans} (8 categorical attributes with domain sizes 51, 36, 15, 8, 6, 5, 4, 3 and 4 numeric attributes with sizes 101, 101, 101, 101) \cite{kaggleloans}.
Following prior work, our first experiment uses the hybrid marginal/prefix workload.
For example, a 1-way hybrid query on attribute $\attr_i$ has the form
``how many records have $\attr_i$ op$_1$ $c_1$'', where op$_1$ is ``='' if $\attr_i$ is categorical
and ``$<$'' if $\attr_i$ is numeric. Then 2-way hybrid queries have the form
``how many records have $\attr_i$ op$_1$ $c_1$ and $\attr_j$ op$_2$ $c_1$'' and so on.
The RMSE of the different methods is shown in Table \ref{tab:prefix-workloads}.

\begin{table}[htbp]
\centering
\caption{RMSE for hybrid marginal/prefix workloads at privacy cost 1.}
\label{tab:prefix-workloads}
\footnotesize
\resizebox{\columnwidth}{!}{%
\begin{tabular}{ll r rrrr}
\hline
Dataset & Workload & \#Queries & RP+ & HDMM & WFF & QS \\
\hline
Adult
  & 1-way       &        588 & 5.11  & 5.648  & 5.804  & \textbf{5.047}  \\
  & 2-way       &    148{,}137 & 17.63 & 21.111 & 22.493 & \textbf{17.632} \\
  & 3-way       & 20{,}894{,}536 & 47.19 & 60.025 & 66.934 & \textbf{47.055} \\
  & $\leq$3-way & 21{,}043{,}261 & 48.90 & 61.613 & 66.794 & \textbf{47.853} \\
\hline
CPS
  & 1-way       &        163 & 3.18 & 3.347 & 3.589 & \textbf{3.135} \\
  & 2-way       &      7{,}000 & 6.36 & 6.625 & 7.841 & \textbf{6.194} \\
  & 3-way       &     72{,}556 & 8.12 & 8.634 & 10.267 & \textbf{7.903} \\
  & $\leq$3-way &     79{,}719 & 8.39 & 8.623 & 10.594 & \textbf{8.140} \\
\hline
Loans
  & 1-way       &        532 & 4.73 & 5.127 & 5.330 & \textbf{4.670} \\
  & 2-way       &    118{,}974 & 14.91 & 17.149 & 18.736 & \textbf{14.822} \\
  & 3-way       & 14{,}539{,}522 & 36.11 & 42.938 & 49.823 & \textbf{36.095} \\
  & $\leq$3-way & 14{,}659{,}028 & 36.65 & 44.338 & 49.721 & \textbf{36.410} \\
\hline
\end{tabular}}
\end{table}
This is a relatively easy workload, with \sysname consistently having the lowest RMSE, RP+ being slightly worse, while HDMM and WFF struggle with 3-way queries on the larger dataset schemas (Adult and Loans).
To accommodate additional  query classes, we consider an experiment where all the attributes are treated as numeric and we evaluate the RMSE on all 2-way \textbf{affine} and 2-way \textbf{Abs} queries. The results are shown in Table \ref{tab:real-aff-abs}.

\begin{table}[htbp]
\centering
\caption{RMSE for 2-way Affine and abs-diff workloads. \sysname improves RMSE by $15-21\%$.}
\label{tab:real-aff-abs}
\footnotesize
\resizebox{\columnwidth}{!}{%
\begin{tabular}{ll r rrrr}
\hline
Dataset & Workload & \#Queries & RP+ & HDMM & WFF & QS \\
\hline
CPS
  & affine      &        805 & 10.012 & NI & 7.147 & \textbf{5.935} \\
  & abs\_diff   &        731 & 10.476 & NI & 7.387 & \textbf{5.900} \\
\hline
Adult
  & affine      &    99{,}830 & 19.395 & NI & 21.559 & \textbf{16.435} \\
  & abs\_diff   &    99{,}200 & 19.441 & NI & 21.587 & \textbf{16.254} \\
\hline
Loans
  & affine      &    66{,}318 & 18.082 & NI & 18.671 & \textbf{14.305} \\
  & abs\_diff   &    65{,}682 & 18.148 & NI & 18.714 & \textbf{14.330} \\
\hline
\end{tabular}}
\end{table}

We see from Table \ref{tab:real-aff-abs} that \sysname consistently outperforms the competing methods
while reducing the noise between $15-21\%$ compared to the next best method. As these are relatively small
datasets, we next consider more systematic evaluations on synthetic data.

\subsection{Utility Experiments for Synthetic Data}\label{sec:exp:synthetic}

We present consider synthetic datasets whose schemas have the form $[n]^d$. That is, there
are $d$ attributes and each attribute has a domain size of $n$. Thus the overall domain
of a tuple is $n^d$. We evaluate all of the workload types as $n$ and $d$ vary.

\begin{table}[htbp]
\centering
\caption{RMSE of 1-way + 2-way workloads at privacy cost$=1$, for different query types, on a synthetic dataset with $40$ attributes ($d=40$) 
as $n$ (domain size of an attribute) varies. NI indicates that competing methods do not have an implementation for the query type. For HDMM, we added
a * to indicate that it can estimate its RMSE but will run out of memory when providing workload query answers.}
\label{tab:2way-varying-n}
\small
\resizebox{\columnwidth}{!}{%
\begin{tabular}{r l r rrrr}
\hline
$n$ & Workload & \#Queries & RP+ & HDMM & WFF & QS \\
\hline
10 & marginal    & 78{,}400 & \textbf{23.48} & \textbf{23.48}${}^*$ & \textbf{23.48} & \textbf{23.48} \\
   & prefix      & 78{,}400 & 36.23 & 52.35${}^*$ &      39.70 &      \textbf{33.70} \\
   & range       & 2{,}361{,}700 & 48.95  & 68.98${}^*$  &    41.36 &    \textbf{41.08} \\
   & circular & 7{,}804{,}000 & NI  & NI  &    \textbf{39.77} &    \textbf{39.77} \\
   & affine      & 15{,}220 & 62.48  & NI  &      45.23 &       \textbf{28.25} \\
   & abs   & 8{,}200 & 85.99  & NI  &      64.11 &       \textbf{35.85} \\
   & random      & 235{,}200 & NI  & NI  &     107.81 &     \textbf{104.43} \\
\hline
20 & marginal    & 312{,}800 & \textbf{25.70} & \textbf{25.70}${}^*$ & \textbf{25.70} & \textbf{25.70} \\
   & prefix      & 312{,}800 & 52.63 & 80.01${}^*$ &     62.95 &     \textbf{49.51} \\
   & range       & 34{,}406{,}400 & 75.44  & 109.56${}^*$  &   63.58 &   \textbf{63.32} \\
   & circular & 124{,}816{,}000 & NI  & NI  &   \textbf{63.01} &   \textbf{63.01} \\
   & affine      & 31{,}220 & 110.30  & NI  &      70.31 &      \textbf{35.71} \\
   & abs   & 16{,}400 & 152.41  & NI  &      102.54 &      \textbf{39.49} \\
   & random      & 938{,}400 & NI  & NI  &   234.89 &   \textbf{233.69} \\
\hline
30 & marginal    & 703{,}200 & \textbf{26.46} & \textbf{26.46}${}^*$ & \textbf{26.46} & \textbf{26.46} \\
   & prefix      & 703{,}200 & 63.10 & 98.35${}^*$ &    79.09 &    \textbf{60.81} \\
   & range       & 168{,}674{,}100 & 93.22  & 137.92${}^*$  &   79.02 &   \textbf{78.79} \\
   & circular & 631{,}836{,}000 & NI  & NI  &  \textbf{79.14} &  \textbf{79.14} \\
   & affine      & 47{,}220 & 149.68  & NI  &     87.62 &     \textbf{44.36} \\
   & abs   & 24{,}600 & 207.22  & NI  &     129.31 &     \textbf{48.14} \\
   & random      & 2{,}109{,}600 & NI  & NI  &   362.86 &   \textbf{362.01} \\
\hline
40 & marginal    & 1{,}249{,}600 & \textbf{26.84} & \textbf{26.84}${}^*$ & \textbf{26.84} & \textbf{26.84} \\
   & prefix      & 1{,}249{,}600 & 70.95 & 112.33${}^*$ &    91.67 &    \textbf{68.78} \\
   & range       & 524{,}504{,}800 & 106.85  & 160.16${}^*$  &  91.11 &  \textbf{90.91} \\
   & circular & 1{,}996{,}864{,}000 & NI & NI & \textbf{91.72} & \textbf{91.72} \\
   & affine      & 63{,}220 & 184.38 & NI &    101.10 &    \textbf{69.62} \\
   & abs   & 32{,}800 & 255.57 & NI &    150.21 &    \textbf{49.83} \\
   & random      & 3{,}748{,}800 & NI & NI &  490.66 &  \textbf{490.65} \\
\hline
50 & marginal    & 1{,}952{,}000 & \textbf{27.07} & \textbf{27.07}${}^*$ & \textbf{27.07} & \textbf{27.07} \\
   & prefix      & 1{,}952{,}000 & 77.29 & 123.73${}^*$ &  102.09 &  \textbf{75.26} \\
   & range       & 1{,}268{,}038{,}500 & 118.00 & 178.65${}^*$ &  101.15 &  \textbf{100.97} \\
   & circular & 4{,}875{,}100{,}000 & NI & NI & \textbf{102.13} & \textbf{102.13} \\
   & affine      & 79{,}220 & 215.98 & NI &  112.23 &  \textbf{79.33} \\
   & abs   & 41{,}000 & 299.62 & NI &  167.51 &  \textbf{52.80} \\
   & random      & 5{,}856{,}000 & NI & NI &  618.61 &  \textbf{616.54} \\
\hline
\end{tabular}}
\end{table}

 Table \ref{tab:2way-varying-n} considers the RMSE of 1-way and 2-way workloads, while 
 Table \ref{tab:mixed-scaling-d-accuracy} considers 1-way + 2-way + 3-way workloads of
 mixed query types. We note that HDMM and RP+ not have implementation support for
 some of the workloads, and we mark such cases in a table cell as NI (not implemented).
 We also note that HDMM cannot fully run on datasets of this scale -- it runs out of memory
 before it can reconstruct noisy workload query answers (since it scales as $O(n^d)$). Nevertheless,
 it can calculate the RMSE it would have in the ideal case of unbounded computational resources. Hence
 we add a * to cell entries where the RMSE is predictable but the algorithm cannot run to completion.

Table \ref{tab:2way-varying-n} shows that \sysname has the lowest RMSE among all methods across all workloads.
Among the competitors, WFF is provably optimal (among matrix mechanisms) for range queries and circular queries \cite{lebeda2025weightedfourierfactorizationsoptimal} and is nearly optimal for range queries.
Meanwhile, RP+ has strong performance for prefix queries (as well as being optimal for marginals \cite{rplus}).
\sysname matches the optimal performance on marginals and circular queries, and is strictly better on all of the
other query types. The largest sources of improvements are for \textbf{Abs} and \textbf{Affine}. 

\begin{table}[htbp]
\centering
\caption{RMSE on mixed-query workloads of
  1-way range queries, 2-way affine queries, 3-way prefix queries
   at privacy cost~1 as $d$ (number of dataset attributes) and $n$ (domain size of each attribute) vary.
}
\label{tab:mixed-scaling-d-accuracy}
\footnotesize
\begin{tabular}{rr rrrr}
\hline
$n$ & $d$ & RP+ & HDMM & WFF & QS \\
\hline
10 & 10 & 22.49 & NI & 25.37 & \textbf{20.41} \\
   & 20 & 57.17 & NI & 65.13 & \textbf{51.63} \\
   & 30 & 99.66 & NI & 113.98 & \textbf{93.50} \\
   & 40 & 148.46 & NI & 170.18 & \textbf{138.38} \\
   & 50 & 202.72 & NI & 232.72 & \textbf{187.24} \\
\hline
20 & 10 & 36.82 & NI & 47.04 & \textbf{34.60} \\
   & 20 & 97.22 & NI & 125.94 & \textbf{95.63} \\
   & 30 & 172.22 & NI & 224.36 & \textbf{167.16} \\
   & 40 & 259.00 & NI & 338.53 & \textbf{249.29} \\
   & 50 & 355.93 & NI & 466.26 & \textbf{340.55} \\
\hline
30 & 10 & 47.08 & NI & 64.37 & \textbf{44.46} \\
   & 20 & 126.20 & NI & 175.18 & \textbf{126.19} \\
   & 30 & 224.99 & NI & 314.28 & \textbf{221.80} \\
   & 40 & 339.60 & NI & 476.12 & \textbf{331.86} \\
   & 50 & 467.86 & NI & 657.56 & \textbf{454.37} \\
\hline
\end{tabular}
\end{table}

Table \ref{tab:mixed-scaling-d-accuracy} presents a more challenging workload
that combines all 1-way range queries, all 2-way affine queries, and all 3-way prefix-sum
queries. As the data dimensionality $d$ and attribute size $n$ vary, \sysname 
shows consistent and large improvements in RMSE.

\subsection{Scalability Experiments}\label{sec:exp:scale}

\sysname has a query decomposition phase, a subworkload optimization phase,
and the assembly phase which computes privacy-preserving queries over the data and
then reconstructs the workload query answers. We next test the scalability of these
phases with results shown in Tables \ref{tab:scaling-nd-timing} and \ref{tab:mixed-scaling-d-timing}.
In these experiments, we consider datasets with $d$ attributes, where each attribute has a domain size $n$ (we vary $d$ and $n$).
Hence, the tuple domain is $n^d$. Table \ref{tab:scaling-nd-timing} considers the easier workload that consists
of all 1-way and 2-way prefix sum queries. Table  \ref{tab:mixed-scaling-d-timing} considers a more challenging
workload that consists of all 1-way range queries, all 2-way affine queries, and all 3-way prefix sum queries.

The tables report the running time in seconds for the decomposition (Decomp.), subworkload optimization (Solve)
and the assembly (Assem.) phases. They also report memory usage. To help understand scaling behavior,
we also included a \textbf{\#Views} column which refers to the number of data
views that would be necessary for answering the workload even in the non-private case. For example, 
when the number of attributes $d=40$, then there are ${40\choose 3}=9880$ types of three-way prefix-sum queries (i.e., $9880$ ways of choosing 3 attributes), ${40\choose 2}=780$ types of two-way affine queries and ${40\choose 1}=40$ types of one-way range queries. Hence $9880+780+40=10700$ dataset views (group-by count queries) on those different attributes are needed to answer the workload even in the non-private case.

In Tables \ref{tab:scaling-nd-timing} and \ref{tab:mixed-scaling-d-timing} we see that the resources (memory and running time) scale well with \textbf{\#Views}. For example, consider the two rows in Table \ref{tab:mixed-scaling-d-timing} corresponding to $n=40, d=10$ and $n=40, d=40$. The number of attributes increases by a factor of 4 but the number of views increases from 175 to 10{,}700 --- a factor of $\approx 61\times$. Meanwhile, the solve time increases from 59.6 seconds to 3031.6 (a factor of $\approx 51\times$) and the memory requirements increase from 3.4\,GB to 227.4\,GB (a factor of $\approx 67\times$).

Table~\ref{tab:scaling-nd-timing} exhibits substantially lower resource usage because it only needs support solve 1-way and 2-way workloads. To get an idea of the necessary memory requirements, suppose $\ell$ queries in a workload can be answered using a marginal with $c$ total cells. A single linear query over that marginal is a linear combination of the values in the $c$ cells and hence is represented as a $c$-dimensional vector. Materializing the workload then requires $\ell c$ space (typically $\ell\geq c$). An optimal matrix mechanism for the part of the workload answerable by this marginal would generally ask for noisy answers to $c$ linear strategy queries, hence requiring $O(c^2)$ space. This is a fundamental bottleneck that is independent of \sysname. The only way around it would be to use sub-optimal matrix mechanisms that are restricted to strategy queries that have a sparse representation and/or to cache optimal solutions (so that marginals with the same shape can re-use the solution). The modular design of \sysname would allow optimal solvers to be replaced with sub-optimal solvers to save resources.

\begin{table}[htbp]
\centering
\caption{Timing experiments on 1-way + 2-way prefix queries as $d$ (number of attributes in the dataset) and $n$ (domain size of each attribute) vary. \#Views refers to the number of data views that would be needed to answer the queries even in the non-private setting.
} 
\label{tab:scaling-nd-timing}
\footnotesize
\resizebox{\columnwidth}{!}{%
\begin{tabular}{rrr rrr r r}
\hline
& & & \multicolumn{3}{c}{\sysname Phase Time} & & \\
\cline{4-6}
$n$ & $d$ &\#Views& Decomp & Solve & Assem & Total & Mem \\
\hline
10 & 10 &55& ${<}$0.1s & 2.8s & ${<}$0.1s & 2.8s & ${<}$1\,GB \\
   & 20 &210& ${<}$0.1s & 4.4s & ${<}$0.1s & 4.4s & ${<}$1\,GB \\
   & 30 &465& ${<}$0.1s & 2.6s & ${<}$0.1s & 2.6s & ${<}$1\,GB \\
   & 40 &820& 0.1s & 7.1s & ${<}$0.1s & 7.2s & ${<}$1\,GB \\
   & 50 &1275& 0.2s & 9.5s & 0.1s & 9.7s & ${<}$1\,GB \\
\hline
20 & 10 &55& ${<}$0.1s & 2.7s & ${<}$0.1s & 2.7s & ${<}$1\,GB \\
   & 20 &210& ${<}$0.1s & 4.2s & ${<}$0.1s & 4.3s & ${<}$1\,GB \\
   & 30 &465& ${<}$0.1s & 6.4s & ${<}$0.1s & 6.4s & ${<}$1\,GB \\
   & 40 &820& 0.1s & 8.1s & ${<}$0.1s & 8.2s & ${<}$1\,GB \\
   & 50 &1275& 0.1s & 10.1s & ${<}$0.1s & 10.2s & 1\,GB \\
\hline
30 & 10 &55& ${<}$0.1s & 3.6s & ${<}$0.1s & 3.6s & 1\,GB \\
   & 20 &210& ${<}$0.1s & 7.0s & ${<}$0.1s & 7.1s & 1\,GB \\
   & 30 &465& 0.1s & 9.9s & ${<}$0.1s & 10.0s & 1\,GB \\
   & 40 &820& 0.1s & 12.7s & ${<}$0.1s & 12.8s & 1\,GB \\
   & 50 &1275& 0.1s & 16.4s & 0.1s & 16.6s & 1\,GB \\
\hline
40 & 10 &55& ${<}$0.1s & 4.9s & ${<}$0.1s & 4.9s & 1\,GB \\
   & 20 &210& ${<}$0.1s & 9.3s & ${<}$0.1s & 9.3s & 1\,GB \\
   & 30 &465& 0.1s & 14.8s & ${<}$0.1s & 14.9s & 1\,GB \\
   & 40 &820& 0.1s & 20.5s & ${<}$0.1s & 20.6s & 1\,GB \\
   & 50 &1275& 0.2s & 37.5s & 0.1s & 37.7s & 2\,GB \\
\hline
\end{tabular}}
\end{table}

\begin{table}[htbp]
\centering
\caption{Timing experiments on mixed workloads containing
  all 1-way range queries, all 2-way affine queries, and all 3-way prefix sum queries,
  as $d$ (number of attributes in the dataset) and $n$ (domain size of each attribute) vary. \#Views refers to the number of data views that would be needed to answer the queries even in the non-private setting.}
\label{tab:mixed-scaling-d-timing}
\footnotesize
\resizebox{\columnwidth}{!}{%
\begin{tabular}{rrr rrr r r}
\hline
& && \multicolumn{3}{c}{\sysname Phase Time} & & \\
\cline{4-6}
$n$ & $d$ &\#Views& Decomp & Solve & Assem & Total & Mem \\
\hline
10 & 10 &175& ${<}$0.1s & 4.5s & ${<}$0.1s & 4.6s & 0.7\,GB \\
   & 20 &1350& 0.5s & 28.1s & 0.8s & 29.4s & 1.0\,GB \\
   & 30 &4525& 1.2s & 25.6s & 2.6s & 29.4s & 2.1\,GB \\
   & 40 &10700& 2.7s & 43.9s & 4.7s & 51.3s & 4.2\,GB \\
   & 50 &20875& 5.5s & 88.7s & 10.0s & 104.2s & 7.7\,GB \\
\hline
20 & 10 &175& 0.1s & 19.8s & 1.1s & 20.9s & 1.0\,GB \\
   & 20 &1350& 0.6s & 46.6s & 4.1s & 51.3s & 3.8\,GB \\
   & 30 &4525& 1.6s & 152.1s & 13.3s & 167.1s & 12.1\,GB \\
   & 40 &10700& 3.9s & 366.2s & 41.3s & 411.4s & 28.4\,GB \\
   & 50 &20875& 7.2s & 751.0s & 79.3s & 837.5s & 55.8\,GB \\
\hline
30 & 10 &175& 0.2s & 22.7s & 2.5s & 25.4s & 1.8\,GB \\
   & 20 &1350& 1.3s & 145.7s & 22.3s & 169.3s & 11.5\,GB \\
   & 30 &4525& 3.4s & 499.2s & 90.0s & 592.5s & 39.1\,GB \\
   & 40 &10700& 7.2s & 1174.0s & 202.7s & 1383.9s & 93.9\,GB \\
   & 50 &20875& 12.6s & 2531.9s & 429.7s & 2974.1s & 186.0\,GB \\
\hline
40 & 10 &175& 0.5s & 59.6s & 8.4s & 68.6s & 3.4\,GB \\
   & 20 &1350& 2.5s & 410.6s & 78.1s & 491.3s & 26.2\,GB \\
   & 30 &4525& 7.5s & 1326.4s & 313.8s & 1647.7s & 91.5\,GB \\
   & 40 &10700& 12.9s & 3031.6s & 650.2s & 3694.7s & 227.4\,GB \\
\hline
\end{tabular}}
\end{table}


\section{Conclusions and Future Work}\label{sec:conc}
In this paper, we introduced \sysname, a matrix mechanism that uses a divide-and-conquer strategy to
answer linear workloads  that are computable from collections of marginals. This
approach allows optimal matrix mechanism constructors to be re-used on small sub-problems. Future research includes improving optimal and approximate low-dimensional matrix mechanisms as the resource savings would automatically transfer to \sysname. Other areas of future  research include  integration  with MWEM style algorithms. Such an integration would combine their respective strengths:  the ability
to plan strategy queries  with data-dependent exploration
of which queries need more accuracy.

\paragraph*{Acknowledgments}
This work was partially supported by NSF awards CNS-2349610, CNS-2317232, and CNS-2317233

\bibliographystyle{ACM-Reference-Format}
\bibliography{refs}

@misc{lebeda2025weightedfourierfactorizationsoptimal,
      title={Weighted Fourier Factorizations: Optimal Gaussian Noise for Differentially Private Marginal and Product Queries}, 
      author={Christian Janos Lebeda and Aleksandar Nikolov and Haohua Tang},
      year={2025},
      eprint={2512.21499},
      archivePrefix={arXiv},
      primaryClass={cs.DS},
      url={https://arxiv.org/abs/2512.21499}, 
}

@article{McKenna_Miklau_Sheldon_2021, 
     title={Winning the NIST Contest: A scalable and general approach to differentially private synthetic data}, 
     volume={11}, 
     number={3}, 
     journal={Journal of Privacy and Confidentiality}, 
     author={McKenna, Ryan and Miklau, Gerome and Sheldon, Daniel}, 
     year={2021}, 
     }

@article{privbayes,
author = {Zhang, Jun and Cormode, Graham and Procopiuc, Cecilia M. and Srivastava, Divesh and Xiao, Xiaokui},
title = {PrivBayes: Private Data Release via Bayesian Networks},
year = {2017},
issue_date = {December 2017},
publisher = {Association for Computing Machinery},
address = {New York, NY, USA},
volume = {42},
number = {4},
journal = {ACM Trans. Database Syst.},
month = {oct},
articleno = {25},
numpages = {41}
}

@inproceedings{aydore2021differentially,
  title={Differentially private query release through adaptive projection},
  author={Aydore, Sergul and Brown, William and Kearns, Michael and Kenthapadi, Krishnaram and Melis, Luca and Roth, Aaron and Siva, Ankit A},
  booktitle={International Conference on Machine Learning},
  pages={457--467},
  year={2021},
  organization={PMLR}
}

@inproceedings{liu2021leveraging,
  title={Leveraging public data for practical private query release},
  author={Liu, Terrance and Vietri, Giuseppe and Steinke, Thomas and Ullman, Jonathan and Wu, Steven},
  booktitle={International Conference on Machine Learning},
  pages={6968--6977},
  year={2021},
  organization={PMLR},
}

@article{hardt2012simple,
  title={A simple and practical algorithm for differentially private data release},
  author={Hardt, Moritz and Ligett, Katrina and McSherry, Frank},
  journal={Advances in neural information processing systems},
  volume={25},
  year={2012}
}

@inproceedings{zhang2018ektelo,
  title={Ektelo: A framework for defining differentially-private computations},
  author={Zhang, Dan and McKenna, Ryan and Kotsogiannis, Ios and Hay, Michael and Machanavajjhala, Ashwin and Miklau, Gerome},
  booktitle={Proceedings of the 2018 International Conference on Management of Data},
  pages={115--130},
  year={2018}
}

@inproceedings{aim,
author = {McKenna, Ryan and Mullins, Brett and Sheldon, Daniel and Miklau, Gerome},
title = {AIM: an adaptive and iterative mechanism for differentially private synthetic data},
year = {2022},
booktitle = {VLDB},
}

@misc{acsdata,
   author = {{U. S. Census Bureau}},
   title = {American Community Survey Data},
   howpublished={\url{https://www.census.gov/programs-surveys/acs/data/summary-file.html}},
}

@misc{census2020,
   author = {{U. S. Census Bureau}},
   title = {Decennial Census Data Products},
   year = {2020},
   howpublished = {\url{https://www.census.gov/programs-surveys/decennial-census/data-products.2020.html}},
}

@inproceedings{dinur:privacy,
author = "Irit Dinur and Kobbi Nissim",
title = "Revealing information while preserving privacy",
booktitle = "PODS",
year = {2003},
}

@article{dawa,
author = {Li, Chao and Hay, Michael and Miklau, Gerome and Wang, Yue},
title = {A data- and workload-aware algorithm for range queries under differential privacy},
year = {2014},
journal = {Proc. VLDB Endow.},
}

@article{mullins2024efficient,
  title={Efficient and private marginal reconstruction with local non-negativity},
  author={Mullins, Brett and Fuentes, Miguel and Xiao, Yingtai and Kifer, Daniel and Musco, Cameron and Sheldon, Daniel R},
  journal={Advances in Neural Information Processing Systems},
  volume={37},
  pages={141356--141389},
  year={2024}
}

@inproceedings{mckenna2025scaling,
title={Scaling up the Banded Matrix Factorization Mechanism for Large Scale Differentially Private {ML}},
author={Ryan McKenna},
booktitle={The Thirteenth International Conference on Learning Representations},
year={2025},
url={https://openreview.net/forum?id=69Fp4dcmJN}
}

@inproceedings{fuentes2026fast,
title={Fast Private Adaptive Query Answering for Large Data Domains},
author={Miguel Fuentes and Brett Mullins and Yingtai Xiao and Daniel Kifer and Cameron N Musco and Daniel Sheldon},
booktitle={The 29th International Conference on Artificial Intelligence and Statistics},
year={2026},
url={https://openreview.net/forum?id=wkYTC6w2pK}
}

@article{hbtree,
 author = {Qardaji, Wahbeh and Yang, Weining and Li, Ninghui},
 title = {Understanding Hierarchical Methods for Differentially Private Histograms},
 journal = {Proc. VLDB Endow.},
 issue_date = {September 2013},
 volume = {6},
 number = {14},
 year = {2013},
}

@inproceedings{YYZH16,
author = {Yuan, Ganzhao and Yang, Yin and Zhang, Zhenjie and Hao, Zhifeng},
title = {Convex Optimization for Linear Query Processing under Approximate Differential Privacy},
year = {2016},
booktitle = {Proceedings of the 22nd ACM SIGKDD International Conference on Knowledge Discovery and Data Mining},
}

@article{YZWXYH12,
author = {Yuan, Ganzhao and Zhang, Zhenjie and Winslett, Marianne and Xiao, Xiaokui and Yang, Yin and Hao, Zhifeng},
title = {Low-Rank Mechanism: Optimizing Batch Queries under Differential Privacy},
year = {2012},
issue_date = {July 2012},
publisher = {VLDB Endowment},
volume = {5},
number = {11},
issn = {2150-8097},
url = {https://doi.org/10.14778/2350229.2350252},
doi = {10.14778/2350229.2350252},
journal = {Proc. VLDB Endow.},
month = jul,
pages = {1352–1363},
numpages = {12},
}

@article{privelet,
author = {Xiao, Xiaokui and Wang, Guozhang and Gehrke, Johannes},
title = {Differential Privacy via Wavelet Transforms},
year = {2011},
journal = {IEEE Transactions on Knowledge and Data Engineering},
volume = {23},
number = {8},
pages = {1200–1214},
}

@article{LMHMR15,
author = {Li, Chao and Miklau, Gerome and Hay, Michael and Mcgregor, Andrew and Rastogi, Vibhor},
title = {The Matrix Mechanism: Optimizing Linear Counting Queries under Differential Privacy},
year = {2015},
issue_date = {December  2015},
publisher = {Springer-Verlag},
address = {Berlin, Heidelberg},
volume = {24},
number = {6},
issn = {1066-8888},
url = {https://doi.org/10.1007/s00778-015-0398-x},
doi = {10.1007/s00778-015-0398-x},
journal = {The VLDB Journal},
month = dec,
pages = {757–781},
numpages = {25}
}

@inproceedings{zcdp,
 author = {Bun, Mark and Steinke, Thomas},
 title = {Concentrated Differential Privacy: Simplifications, Extensions, and Lower Bounds},
 booktitle = {Proceedings, Part I, of the 14th International Conference on Theory of Cryptography - Volume 9985},
 year = {2016},
}

@article{commonmech,
author = {Xiao, Yingtai and Wang, Guanhong and Zhang, Danfeng and Kifer, Daniel},
title = {Answering Private Linear Queries Adaptively Using the Common Mechanism},
year = {2023},
issue_date = {April 2023},
publisher = {VLDB Endowment},
volume = {16},
number = {8},
issn = {2150-8097},
url = {https://doi.org/10.14778/3594512.3594519},
doi = {10.14778/3594512.3594519},
journal = {Proc. VLDB Endow.},
month = apr,
pages = {1883–1896},
numpages = {14}
}

@inproceedings{renyidp,
  author    = {Ilya Mironov},
  title     = {R{\'{e}}nyi Differential Privacy},
  booktitle = {30th {IEEE} Computer Security Foundations Symposium, {CSF} 2017, Santa
               Barbara, CA, USA, August 21-25, 2017},
  pages     = {263--275},
  year      = {2017},
 }

@inproceedings{dwork06Calibrating,
  author    = {Cynthia Dwork and
               Frank McSherry and
               Kobbi Nissim and
               Adam Smith},
  title     = {Calibrating Noise to Sensitivity in Private Data Analysis.},
  booktitle = {TCC},
  year      = {2006},
}

@INPROCEEDINGS{dworkKMM06:ourdata,
 AUTHOR = {Cynthia Dwork and Krishnaram Kenthapadi and Frank McSherry and Ilya Mironov and Moni Naor},
 TITLE = {Our Data, Ourselves: Privacy via Distributed Noise Generation},
 BOOKTITLE = {EUROCRYPT},
 YEAR = {2006},
 PAGES = {486--503},
}

@article{tdahdsr,
  title = {The 2020 Census Disclosure Avoidance System TopDown Algorithm},
  author = {John M. Abowd and Robert Ashmead and  Ryan Cumings-Menon and Simson Garfinkel and
Micah Heineck and  Christine Heiss and Robert Johns and Daniel Kifer and Philip Leclerc and Ashwin Machanavajjhala and Brett Moran and  William Sexton and 
Matthew Spence and Pavel Zhuravlev},
  journal = {Harvard Data Science Review},
  year = {forthcoming. Preprint \url{https://www.census.gov/library/working-papers/2022/adrm/CED-WP-2022-002.html}},
}

@article{fdp,
author = {Dong, Jinshuo and Roth, Aaron and Su, Weijie J.},
title = {Gaussian differential privacy},
journal = {Journal of the Royal Statistical Society: Series B (Statistical Methodology)},
volume = {84},
number = {1},
pages = {3-37},
doi = {https://doi.org/10.1111/rssb.12454},
url = {https://rss.onlinelibrary.wiley.com/doi/abs/10.1111/rssb.12454},
eprint = {https://rss.onlinelibrary.wiley.com/doi/pdf/10.1111/rssb.12454},
year = {2022}
}

@inproceedings{gaboardi2014dual,
  title={Dual query: Practical private query release for high dimensional data},
  author={Gaboardi, Marco and Arias, Emilio Jes{\'u}s Gallego and Hsu, Justin and Roth, Aaron and Wu, Zhiwei Steven},
  booktitle={International Conference on Machine Learning},
  pages={1170--1178},
  year={2014},
  organization={PMLR}
}

@article{hay2009boosting,
  title={Boosting the Accuracy of Differentially Private Histograms Through Consistency},
  author={Hay, Michael and Rastogi, Vibhor and Miklau, Gerome and Suciu, Dan},
  journal={Proceedings of the VLDB Endowment},
  volume={3},
  number={1},
  year={2010}
}

@misc{qcewdata,
  title = {{QCEW} Public-Use Data Files},
  author = {{U.S. Bureau of Labor Statistics}},
  howpublished={\url{https://www.bls.gov/cew/downloadable-data-files.htm}},
year={(retrieved August 21, 2024)}
}

@article{mckenna2018optimizing,
  title={Optimizing error of high-dimensional statistical queries under differential privacy},
  author={McKenna, Ryan and Miklau, Gerome and Hay, Michael and Machanavajjhala, Ashwin},
  journal={Proceedings of the VLDB Endowment},
  volume={11},
  number={10},
  year={2018}
}

@inproceedings{rplanner,
author = {Xiao, Yingtai and He, Guanlin and Zhang, Danfeng and Kifer, Daniel},
title = {An optimal and scalable matrix mechanism for noisy marginals under convex loss functions},
year = {2023},
publisher = {Curran Associates Inc.},
address = {Red Hook, NY, USA},
booktitle = {Proceedings of the 37th International Conference on Neural Information Processing Systems},
articleno = {901},
numpages = {45},
location = {New Orleans, LA, USA},
series = {NeurIPS},
}

@misc{rplus,
      title={ResidualPlanner+: a scalable matrix mechanism for marginals and beyond}, 
      author={Yingtai Xiao and Guanlin He and Levent Toksoz and Zeyu Ding and Danfeng Zhang and Daniel Kifer},
      year={2025},
      eprint={2305.08175v3},
      archivePrefix={arXiv},
      primaryClass={cs.DB},
      url={https://arxiv.org/abs/2305.08175v3}, 
}

@article{yuan2015optimizing,
  title={Optimizing batch linear queries under exact and approximate differential privacy},
  author={Yuan, Ganzhao and Zhang, Zhenjie and Winslett, Marianne and Xiao, Xiaokui and Yang, Yin and Hao, Zhifeng},
  journal={ACM Transactions on Database Systems (TODS)},
  volume={40},
  number={2},
  pages={1--47},
  year={2015},
  publisher={ACM New York, NY, USA}
}

@inproceedings{edmonds2020power,
  title={The power of factorization mechanisms in local and central differential privacy},
  author={Edmonds, Alexander and Nikolov, Aleksandar and Ullman, Jonathan},
  booktitle={Proceedings of the 52nd Annual ACM SIGACT Symposium on Theory of Computing},
  pages={425--438},
  year={2020}
}

@phdthesis{nikolov2014new,
  title={New computational aspects of discrepancy theory},
  author={Nikolov, Aleksandar},
  year={2014},
  school={Rutgers University-Graduate School-New Brunswick}
}

@inproceedings{xiao2020optimizing,
  title={Optimizing fitness-for-use of differentially private linear queries},
  author={Xiao, Yingtai and Ding, Zeyu and Wang, Yuxin and Zhang, Danfeng and Kifer, Daniel},
  booktitle={VLDB},
  year={2021}
}

@misc{kaggleloans,
  author = "Kaggle",
  title = "Loan Prediction Problem Dataset",
  howpublished = "\url{https://www.kaggle.com/altruistdelhite04/loan-prediction-problem-dataset}",
  year = "2021",
  note = "Accessed: May 8th, 2023"
}

@misc{qsarxiv,
   year = {2026},
   title = {Accurate and Scalable Matrix Mechanisms via Divide and Conquer},
   author = {Guanlin He and Yingtai Xiao and Jiamu Bai and Xin Gu and Zeyu Ding  and Wenpeng Yin and Daniel Kifer},
   howpublished = {Arxiv},
}

@misc{adult_2,
  author       = {Becker, Barry and Kohavi, Ronny},
  title        = {{Adult}},
  year         = {1996},
  howpublished = {UCI Machine Learning Repository},
  note         = {{DOI}: https://doi.org/10.24432/C5XW20}
}

\ifarxiv
\clearpage
\appendix

\section{Proofs from Section {\protect \ref{sec:decomp}}}\label{app:decomposition}
\printProofs[decomp]

\section{Proofs from Section {\protect \ref{sec:accuracy}}}\label{app:accuracy}
\printProofs[accuracy]

\section{Integration of Fourier Basis into {\protect \sysname}}\label{app:fourier}
In this section, we explain how to integrate the 
 Fourier basis \cite{lebeda2025weightedfourierfactorizationsoptimal} into \sysname 
 so that it can be used as a fast, but approximate, method to
 optimize a subworkload $\workload_{\deco\attset^\prime}$. If the Fourier basis is used for all subworkloads, the result becomes equivalent to
 WFF.
 
Following the ideas of Lebeda et al. \cite{lebeda2025weightedfourierfactorizationsoptimal}, the high level idea is to 
 take the multidimensional Fourier transform of the tensor view of a marginal $\marginal_{\attset^\prime}$, add different amounts
 of noise to each Fourier coefficient, and then take the inverse (multidimensional) Fourier transform
 of the noisy coefficients. This results in a privacy-preserving marginal $\widehat{\marginal}_{\attset^\prime}$
 that can be used to directly answer queries in the subworkload. Thus, the goal is to determine how much noise different coefficients need.

\subsubsection*{Expressing query variances in terms of tunable parameters.}
A marginal  ${\marginal}_{\attset^\prime}$, where $\attset^\prime=\{\attr_{i_1},\dots,\attr_{i_\ell}\}$, can be indexed using a tuple of $(j_1,\dots, j_\ell)$ of $\ell$ indexes (one for each dimension), where $j_1$ range from $0$ to $\attsize_{i_1}-1$ (inclusive), the index $j_2$ ranges from 0 to $\attsize_{i_2}-1$, etc. The parameter set $\theta$ can be thought of as a similarly sized table, but only some entries correspond to actual parameters. Specifically, $\theta[j_1,\dots, j_\ell]$ is a tunable parameter if $j_1 \in [1,\dots, \attsize_{i_1}-1], \dots, j_\ell\in [1,\dots,\attsize_{i_\ell}-1]$ and the tuple $(j_1,\dots, j_\ell)$ is $\leq$ than the tuple $(\attsize_{i_1}-j_1, \attsize_{i_2}-j_2, \dots, \attsize_{i_\ell}-j_\ell)$ in dictionary order.

Given a query $\query\in \workload_{\deco\attset^\prime}$ expressed in tensor form, its data-independent variance in terms of the tunable parameters can be computed as follows:
\begin{enumerate}
    \item Let $\widetilde{F}$ be the inverse multidimensional discrete Fourier transform of $\query$.
    \item The variance is a linear combination of the parameters in $\theta$, where the coefficient of the parameter $\theta[j_1,\dots, j_\ell]$ is
    $||\widetilde{F}[j_1,\dots, j_\ell]||^2$ if  $(j_1,\dots, j_\ell) =(\attsize_{i_1}-j_1, \attsize_{i_2}-j_2, \dots, \attsize_{i_\ell}-j_\ell)$ and
       $ 4 ||\widetilde{F}[j_1,\dots, j_\ell]||^2$otherwise.
\end{enumerate}
Furthermore, if $\query^\prime$ is any circular-shift of $\query$, then they have the same variance.

\subsubsection*{Setting the tunable parameters.}
The weighted sum of variances of queries in $\workload_{\deco\attset^\prime}$ can be computed by multiplying each query's variance by the query's weight and adding them up. Thus,
the weighted sum of variances is a linear combination of the parameters in $\theta$. Let $c_{j_1,\dots, j_\ell}$ be the coefficient of parameter $\theta[j_1,\dots, j_\ell]$. Let $\gamma$ be the sum of the square roots of the linear coefficients. Then the value of
parameter is  $\theta[j_1,\dots, j_\ell]=\gamma/\sqrt{c_{j_1,\dots, j_\ell}}$.

\subsubsection*{How to run the mechanism.} Setting the values for the parameter set $\theta$ defines a mechanism. This mechanism can be run as follows:
\begin{enumerate}
    \item Create a tensor $Z$ with entries initialized to 0.
    \item Let $F$ be the multidimensional discrete Fourier transform of $\marginal_{\attset^\prime}$.
    \item For each valid parameter index $(j_1,\dots, j_\ell)$, let $\sigma^2=\theta[j_1,\dots, j_\ell]$ and let $a+b\imag=F[j_1,\dots, j_\ell]$.
    \begin{enumerate}
        \item Add noise: $\widetilde{a}=a + N(0, \sigma^2)$ and $\widetilde{b} = b + N(0, \sigma^2)$.
        \item Store: set $Z[j_1, \dots, j_\ell]= \widetilde{a} + \widetilde{b}\imag$ and $Z[\attsize_{i_1}-j_1, \dots, \attsize_{i_\ell}-j_\ell]= \widetilde{a} - \widetilde{b}\imag$.
    \end{enumerate}
    \item Set $\widetilde{\marginal}_{\attset^\prime}$ to be the inverse multidimensional discrete Fourier transform of $Z$ and return $\widetilde{\marginal}_{\attset^\prime}$.
\end{enumerate}

\subsubsection*{How to reconstruct a query answer.}
Any query $\query\in\workload_{\deco\attset^\prime}$ can be directly answered from  $\widetilde{\marginal}_{\attset^\prime}$ in the usual way: perform element-wise multiplication of the tensor view of $\query$ with $\widetilde{\marginal}_{\attset^\prime}$ and then add up all the entries.

\section{Variations of the Matrix Mechanism Optimization Problem}\label{app:variations}

We present additional variations of the matrix mechanism construction optimization problem from Section  \ref{sec:solvers}, when the goal is to minimize weighted sum squared error.
Let $\mat{W}$ be the matrix whose rows are the query vectors in $\workload_{\deco\attset^\prime}$.
Let $\mat{D}$ be a diagonal matrix, where entry $\mat{D}[i,i]$ is the weight of the query that is in the $i^\text{th}$ row of $\mat{W}$

Suppose we have an independent linear basis $\mat{U}$ whose row span is equal to the row span of $\mat{W}$ and $\mat{U}^\dagger$ is the pseudoinverse of $\mat{U}$. Then one can solve the following optimization problem:

\begin{align*}
    \arg\min_{\mat{V}}\quad&\text{trace}((\mat{U}^\dagger)^\top\mat{W}^\top\mat{D}\mat{W}\mat{U}^\dagger\mat{V}^{-1})\\
    &\text{s.t. }\mat{V} \text{ is symmetric and invertible }\\
    &\text{and }(\mat{U}^\top\mat{V}\mat{U})[i,i]\leq 1 \text{ for all }i
\end{align*}
This is a semidefinite program that can be solved using off-the-shelf optimizers.
Once a solution is obtained, $\bmat$ can be computed by taking the Cholesky decomposition $\mat{V}=\mat{L}^\top\mat{L}$ and setting $\bmat=\mat{L}\mat{U}$. The rest is the same as in Section \ref{sec:solvers}.

Some solvers for the optimization  problem in Section \ref{sec:solvers} require $\mat{W}$ to be full-rank \cite{yuan2015optimizing}. However, the $\mat{W}$ we get from a workload partition $\workload_{\deco\attset^\prime}$ is typically not full rank. For example, in Figure \ref{fig:decomposition_post} the query vectors all have entries that add up to 0. In the tensor view of each query, the sum along any dimension is clearly 0 (this is a direct consequence of the query decomposition construction).
In fact, the rows of the query matrix $\mat{W}$ for the workload $\workload_{\deco\attset^\prime}$, where $\attset^\prime=\{\attr_{i_1},\dots,\attr_{i_\ell}\}$, 
always belong to the subspace spanned by the rows of 
\begin{align}
   \left(\identity_{\attsize_{i_1}} - \frac{1}{\attsize_{i_1}}\one_{\attsize_{i_1}}\one_{\attsize_{i_1}}^T\right)
   \kron \cdots \kron 
      \left(\identity_{\attsize_{i_\ell}} - \frac{1}{\attsize_{i_\ell}}\one_{\attsize_{i_\ell}}\one_{\attsize_{i_\ell}}^T\right)\label{eqn:queryspace}
\end{align}
To sidestep this issue of a non-full-rank matrix, \sysname can be modified in the following way:
\begin{enumerate}[leftmargin=*]
\item If \sysname is being used with a matrix mechanism constructor that needs a full rank query workload, one can change how  $\workload_{\deco\attset^\prime}$ is constructed. Given a query that is answerable from a marginal on $\attset$  (with $\attset^\prime\subseteq \attset$), the decomposition in Equation \ref{eq:vecdecomp} would use the matrices $\identity_{\attsize_{i_j}}$ instead of $\identity_{\attsize_{i_j}} - \frac{1}{\attsize_{i_j}}\one_{\attsize_{i_j}}\one_{\attsize_{i_j}}^T$.
\item After the convex matrix mechanism constructor has returned a $\bmat_{\text{tmp}}$, we project its rows onto the subspace spanned by the rows of Equation \ref{eqn:queryspace} by setting $\bmat^\prime$ to be $\bmat_{\text{tmp}}$ times the matrix in Equation \ref{eqn:queryspace}.
\item The rest is the same as in Section \ref{sec:solvers}.
\item The reconstruction phase of Algorithm \ref{alg:smasher} uses the original query decomposition.
\end{enumerate}


\section{Evaluation on Weighted Queries}\label{app:weighted}
\vspace{0.5em}\noindent\textbf{Workload weights.}
We also experiment with \textbf{weighted} workloads to see if the results qualitatively change.
 For weighted workloads, we randomly sample an integer weight
for each query from 1 to 5, to simulate the situation where some queries are more important than others.
The weighted RMSE is defined as follows:

\begin{align*}
    WRMSE &=\sqrt{\frac{\sum\limits_{\queryv\in\workload}\weight(\queryv)\varfun(\queryv)}{|\workload|}}
\end{align*}

Table \ref{tab:2way-varying-n_weighted} shows the results. Qualitatively, little changes from Table \ref{tab:2way-varying-n},
indicating that the latest scalable matrix mechanisms may be robust to weighted queries.

\begin{table}[htbp]
\centering
\caption{WRMSE at privacy cost 1 of weighted 1-way and 2-way marginal workload as $d$ (dataset dimensions) and $n$ (domain size of each attribute) vary. on $[n]^{40}$, varying domain size~$n$.
  The RP+ and HDMM implementations do not support these queries.}
\label{tab:2way-varying-n_weighted}
\small
\begin{tabular}{r l r rr}
\hline
$n$ & Workload & \#Queries & WFF & QS \\
\hline
10 & marginal    & 78{,}400 & \textbf{39.22} & \textbf{39.22} \\
   & prefix      & 78{,}400 & 66.58 & \textbf{56.59} \\
   & range       & 2{,}361{,}700 & 69.25 & \textbf{68.82} \\
   & circular & 7{,}804{,}000 & \textbf{66.70} & \textbf{66.70} \\
   & affine      & 15{,}220 & 75.74 & \textbf{47.33} \\
   & abs   & 8{,}200 & 107.17 & \textbf{51.94} \\
   & random      & 235{,}200 & 180.03 & \textbf{174.49} \\
\hline
20 & marginal    & 312{,}800 & \textbf{42.90} & \textbf{42.90} \\
   & prefix      & 312{,}800 & 105.43 & \textbf{84.18} \\
   & range       & 34{,}406{,}400 & 106.37 & \textbf{105.93} \\
   & circular & 124{,}816{,}000 & \textbf{105.53} & \textbf{105.53} \\
   & affine      & 31{,}220 & 117.65 & \textbf{61.68} \\
   & abs   & 16{,}400 & 171.30 & \textbf{66.89} \\
   & random      & 938{,}400 & 391.96 & \textbf{386.32} \\
\hline
30 & marginal    & 703{,}200 & \textbf{44.16} & \textbf{44.16} \\
   & prefix      & 703{,}200 & 132.39 & \textbf{101.95} \\
   & range       & 168{,}674{,}100 & 132.15 & \textbf{131.76} \\
   & circular & 631{,}836{,}000 & \textbf{132.48} & \textbf{132.48} \\
   & affine      & 47{,}220 & 146.56 & \textbf{102.25} \\
   & abs   & 24{,}600 & 215.96 & \textbf{106.11} \\
   & random      & 2{,}109{,}600 & \multicolumn{2}{c}{\emph{not finished}} \\
\hline
40 & marginal    & 1{,}249{,}600 & \textbf{44.80} & \textbf{44.80} \\
   & prefix      & 1{,}249{,}600 & 153.42 & \textbf{115.07} \\
   & range       & 524{,}504{,}800 & 152.33 & \textbf{152.00} \\
   & circular & 1{,}996{,}864{,}000 & \textbf{153.50} & \textbf{153.50} \\
   & affine      & 63{,}220 & 169.06 & \textbf{139.22} \\
   & abs   & 32{,}800 & 250.83 & \textbf{183.22} \\
   & random      & 3{,}748{,}800 & 818.49 & \textbf{817.66} \\
\end{tabular}
\end{table}

\fi
\end{document}